\documentclass[journal,twoside,web]{ieeecolor}

\usepackage{generic}
\usepackage{cite}
\usepackage{amsmath,amssymb,amsfonts}
\usepackage{algorithmic}
\usepackage{graphicx}
\usepackage{algorithm,algorithmic}
\usepackage{hyperref}
\hypersetup{hidelinks}
\usepackage{textcomp}
\let\labelindent\relax
\usepackage{enumitem}
\usepackage{bm}
\usepackage{xcolor}
\usepackage{ifthen}

\newboolean{TAC}
\setboolean{TAC}{false}

\def\BibTeX{{\rm B\kern-.05em{\sc i\kern-.025em b}\kern-.08em
    T\kern-.1667em\lower.7ex\hbox{E}\kern-.125emX}}
\newtheorem{thm}{\it Theorem}
\newtheorem{lem}{\it Lemma}
\newtheorem{assumpt}{\it Assumption}
\newtheorem{problem}{\it Problem}
\newtheorem{remark}{\it Remark}
\newtheorem{corol}{\it Corollary}
\newtheorem{prop}{\it Proposition}
\newtheorem{define}{\it Definition}
\newtheorem{example}{\it Example}

\title{Output Regulation of Linear Systems with Non-periodic Non-smooth Exogenous Signals}
\author{Zirui Niu, \IEEEmembership{Graduate Student Member, IEEE}, Daniele Astolfi, and Giordano Scarciotti, \IEEEmembership{Senior Member, IEEE}
\thanks{Zirui Niu and G. Scarciotti is with the Department of Electrical and Electronic Engineering, Imperial College London, London SW7 2AZ, U.K. (e-mail: zn120@ic.ac.uk; g.scarciotti@ic.ac.uk).}
\thanks{D. Astolfi is with Universit\'e Claude Bernard Lyon 1, CNRS, LAGEPP UMR5007, 43 boulevard du 11 novembre 1918, F-69100, Villeurbanne, France (e-mail: 
daniele.astolfi@univ-lyon1.fr). }  \vspace{-0.3cm}}

\begin{document}

\maketitle

\begin{abstract}
We address the output regulation problem of linear systems with non-smooth and non-periodic exogenous signals. Specifically, we first formulate and solve the full-information problem by designing a state-feedback controller. We study the solvability of the regulator equations, providing a new non-resonance condition.
We then focus on the error-feedback problem, for which we design a (non-robust) internal model leveraging the concept of canonical realisation and applying a high-gain method for the stabilisation of the closed-loop system under the minimum-phase assumption. Finally, we study the regulation problem involving
model parameter uncertainties. Drawing ideas from both hybrid and time-varying (smooth) output regulation, we propose two methods to establish an internal model that is robust to uncertainties. The first method is an extension of the hybrid internal model, while the second relies on a new concept of immersion. In this non-smooth case, the immersion is established based on integrals rather than derivatives. The effectiveness of the proposed solutions is illustrated by a circuit regulation example. 
\end{abstract}


\section{Introduction}\label{sec:intro}
\IEEEPARstart{I}{n} the field of control theory, designing a controller such that the output of a system can asymptotically track/reject specified references/disturbances is a fundamental problem. In this context, output regulation theory models references and disturbances, called ``exogenous signals", as trajectories generated by a known signal generator named ``exosystem". Research into output regulation problems for linear systems dates back to 1970s, when Davison, Francis and others studied the servomechanism problem using elegant geometric methods~\cite{ref:davison1976robust, ref:francis1976internal, ref:francis1977linear, ref:wonham1974linear}. They demonstrated that a robust solution capable of handling parametric and dynamic uncertainties in the regulated system requires the regulator to incorporate an appropriate replica of the exosystem dynamics, leading to the \textit{internal model principle}. Since the 1990s, the conceptual tools and design principles of the output regulation theory have been extended to nonlinear systems~\cite{ref:isidori1990output, ref:huang1990nonlinear, ref:byrnes1997structurally, ref:serrani2001semi, ref:chen2005robust, ref:marconi2007output}. Recently, many studies focusing on output regulation of other classes of dynamical systems have been presented. To mention a few, multi-agent systems have gained attention since 2010~\cite{ref:su2011cooperative, ref:su2012cooperative};
hybrid systems have been studied during the same time~\cite{ref:marconi2013internal, ref:carnevale2015hybrid, ref:carnevale2017robust}; and regulation methods of linear stochastic systems have also been reported, see~\cite{ref:scarciotti2018output, ref:mellone2021output}.

In this paper, we are interested in non-periodic non-smooth exogenous signals. Non-smooth signals, such as sawtooth or pulse width modulation (PWM) signals, are commonly encountered in real-life applications. For example, discontinuous or non-differentiable signals, possibly non-periodic, appear in robotic manipulation, either as disturbances, as references, or due to the interconnection between agents and the environment~\cite{ref:sira1996dynamical, ref:cortes2008discontinuous, ref:chen2013tracking, ref:akhmet2010principles}. This class of exogenous signals has not been considered by past output regulation studies. On the one hand, many works have focused on smooth exogenous signals represented by a known, autonomous differential equation, which herein we call \textit{implicit} generator~\cite{ref:huang2004nonlinear}. This generator, in time-invariant settings, cannot model signals that show non-smoothness. Similarly, for these standard implicit generators even in time-varying settings, the existing papers have only considered periodic and smooth exogenous signals~\cite{ref:zhang2006linear,ref:zhang2009adaptive}. On the other hand, tracking and rejection of non-smooth signals have been studied in more general contexts. For instance,~\cite{ref:gazi2007output} studied linear output regulation problems with switching exosystems. The study only considered continuous non-smooth exogenous signals and solved the problem with a switching controller, while our study also considers discontinuous exogenous signals.~\cite{ref:paunonen2017robust} studied, using the lifting approach, the regulation of linear time-varying systems subject to periodic exogenous signals that can be non-smooth. Other studies that considered non-smooth exogenous signals are mostly based on hybrid linear systems, see~\cite{ref:marconi2013internal, ref:carnevale2015hybrid, ref:carnevale2017robust, ref:de2019robust}. All these works solve the non-smooth problem for exogenous signals having \textit{periodic} jumps. However, output regulation with \textit{non-periodic}, possibly discontinuous, exogenous signals still remains an open question.

This article addresses the output regulation problem for single-input single-output (SISO), linear time-invariant (LTI) systems with \textit{non-periodic non-smooth} exogenous signals. To this end, these exogenous signals are modelled by a linear exosystem represented in explicit form\footnote{The terminology is taken from~\cite{ref:zadeh2008linear, ref:kalman1969topics}.} as~\cite[Section 5.1]{ref:scarciotti2015model, ref:scarciotti2017nonlinear}
\begin{equation}\label{equ:explicitGen}
\omega(t) = \Lambda\left(t, t_0\right) \omega(t_0), \qquad  \omega(t_0) = \omega_0,
\end{equation}
with $\Lambda\left(t, t_0\right) \in \mathbb{R}^{\nu \times \nu}$. This class of exosystems can represent most signals with discontinuities or non-differentiabilities. As a consequence, these exosystems bring more challenges in solving the output regulation problem as they pose restrictions in the use of derivatives. Preliminary results have been published in~\cite{ref:niu2024full}, which studied the full-information\footnote{This means that the states of the linear system and the exogenous signal are available for feedback.} problem with \textit{bounded} exogenous signals generated by~(\ref{equ:explicitGen}) and proposed solutions in the form of \textit{regulator equations}.

\textit{Contributions.} In this article, we first revisit this full-information problem, providing complete proofs of all theoretical contributions while introducing a major improvement by including the \textit{unbounded} case (Section~\ref{sec:FullInfoProb}). 
We then examine the solvability of the regulator equations, providing a new non-resonance condition (Section~\ref{sec:solvability}).
Based on these regulator equations, we then solve the error-feedback case by presenting necessary and sufficient conditions, called \textit{internal model property}, for the regulator design (Section~\ref{sec:IMProp}). By generalising the canonical realization used in the smooth periodic case~\cite{ref:zhang2009adaptive} (Section~\ref{sec:AchievingIMProp}), we provide a constructive procedure for designing an error-feedback regulator to solve the problem for minimum-phase systems (Section~\ref{sec-robstab}). Then, we seek to extend the posed internal model property to the \textit{internal model principle} when the system presents parameter uncertainties. To design a robust regulator, we propose two alternative strategies. We first extend the hybrid internal model principle posed by~\cite{ref:marconi2013internal} in the periodic setting to the non-periodic case, showing that the method is viable, but with a drawback of augmented dimension of the regulator (Section~\ref{sec:HybridIMP}). Then, we provide an alternative constructive method of designing the robust regulator via a new concept of \textit{immersion}. Differently from the classical immersion method which depends on the successive derivatives of the expected error-zeroing input signal~\cite{ref:isidori1995nonlinear}, our proposed method is integral-based to cope with the non-smooth case (Section~\ref{sec:ImmersionIMP}). We finally compare the two methods, demonstrating that both methods have a similar level of restrictiveness, with the immersion-based regulator showing an advantage of lower regulator dimensionality in some cases. This comparison is also illustrated by an example at the end of the article (Section~\ref{sec:example}).


\textbf{Notation.} $\mathbb{R}_{\geq 0}$ denotes the set of non-negative real numbers, $\mathbb{R}_{>0}$ indicates $\mathbb{R}_{\geq 0} \backslash\{0\}$, 
$\mathbb{C}_{<0}$ denotes the set of complex numbers with a strictly negative real part and $\mathbb{C}_{\geq 0}$ denotes $\mathbb{C} \backslash \mathbb{C}_{<0}$. The symbol $\otimes$ indicates the Kronecker product, $I_{n}$ denotes an $n \times n$ identity matrix, $\bm{0}_{m \times n}$ denotes an $m \times n$ zero matrix, $\sigma(A)$ indicates the spectrum of the matrix $A \in \mathbb{R}^{n \times n}$ and $\|A\|$ indicates its induced Euclidean matrix norm. The superscript $\top$ denotes the transposition operator, 
and the superscript $\dagger$ denotes the Moore-Penrose pseudoinverse operator. For a matrix $V \in \mathbb{R}^{a \times b}$, $\operatorname{vec}(V)$ indicates an $ab \times 1$ column vector obtained by vertically stacking the columns of $V$. Meanwhile, for $n$ matrices $M_{i} \in \mathbb{R}^{p \times q}$, $i = 1, 2, \cdots, n$, $\operatorname{col}\left(M_1, M_2, \cdots, M_n\right)$ indicates the $pn \times q$ matrix obtained by vertically stacking the matrices $M_i$. When the matrices $M_{i}$ are square, $\operatorname{diag}(M_1, M_2 \cdots, M_n)$ denotes the square matrix obtained by placing the matrices $M_1, M_2 \cdots, M_n$ along its diagonal. Given a function $H(\cdot)$ and a positive integer $k$, the notation $H \in \mathcal{C}^{k}$ indicates that $H$ is $k$-times differentiable with its $k$-th derivative, denoted by $H^{(k)}$, being continuous in its domain. On the other hand, given an initial value $t_0$, the functional $\mathcal{I}^{[k]}_{t_0} :  H \mapsto \mathcal{I}^{[k]}_{t_0}[H]$ denotes $k$-times repeated integral of the function $H$, namely
$$
\mathcal{I}^{[k]}_{t_0}[H(t)] =\int_{t_0}^t \int_{t_0}^{\sigma_1} \cdots \int_{t_0}^{\sigma_{k-1}} H\left(\sigma_k\right) \mathrm{d} \sigma_{k} \cdots \mathrm{d} \sigma_2 \mathrm{~d} \sigma_1.
$$

\section{Preliminaries}\label{sec:prelimiary}
Consider a class of SISO, LTI systems in the form
\begin{equation}\label{equ:system}
    \begin{aligned}
    \dot{x}(t) &= Ax(t) + Bu(t) + P\omega(t),\\
    e(t) &= Cx(t) + Du(t) + Q\omega(t),
    \end{aligned}
\end{equation}
with $A \in \mathbb{R}^{n \times n}$, $B \in \mathbb{R}^{n \times 1}$, $P \in \mathbb{R}^{n \times \nu}$, $C \in \mathbb{R}^{1 \times n}$, $D \in \mathbb{R}$, $Q \in \mathbb{R}^{1 \times \nu}$, $x(t) \in \mathbb{R}^n$ the state, $u(t) \in \mathbb{R}$ the control input, $e(t) \in \mathbb{R}$ the regulation error, and $\omega(t) = [r(t),\; d(t)^{\top}]^{\top} \in \mathbb{R}^{v}$ the exogenous signal representing the disturbances $d(t) \in \mathbb{R}^{\nu-1}$ and/or reference signal $r(t) \in \mathbb{R}$. 

The classical linear output regulation problem considers the exogenous signals as the solutions of systems in implicit form~\cite{ref:huang2004nonlinear}
\begin{equation}\label{equ:implicitGen}
\dot{\omega}(t)=S \omega(t), \qquad \omega(t_0) = \omega_0,
\end{equation}
where $S \in \mathbb{R}^{\nu \times \nu}$. As mentioned in Section~\ref{sec:intro}, since this model cannot generate signals that are not differentiable, we study non-periodic non-smooth exogenous signals by expressing $\omega$ as the solution of a generator in explicit form~(\ref{equ:explicitGen}), which, compared with~(\ref{equ:implicitGen}), can represent a wider range of signals. For instance, exosystem~(\ref{equ:explicitGen}) produces all signals generated by system~(\ref{equ:implicitGen}) with $\Lambda(t,t_0) = e^{S(t-t_0)}$, and even all signals generated by a time-varying system of the form
\begin{equation}\label{equ:timeVaryingGen}
\dot{\omega}(t)= \tilde{S}(t) \omega(t), \qquad \omega(t_0) = \omega_0,
\end{equation}
with $\tilde{S}(t) \in \mathbb{R}^{\nu \times \nu}$. In this time-varying case, $\Lambda\left(t, t_0\right)$ is the transition matrix associated to~(\ref{equ:timeVaryingGen})~\cite[Section 3]{ref:brockett2015finite}. Moreover, generator~(\ref{equ:explicitGen}) can also express signals provided by possibly time-varying hybrid systems of the form
\begin{equation}\label{equ:hybridGen}
    \begin{aligned}
    \dot{\omega}(t, k) = \bar{S}(t, k) \omega(t, k), \quad
    \omega(t, k+1) = J(t, k) \omega(t, k),
    \end{aligned}
\end{equation}
where $\bar{S}(t, k) \in \mathbb{R}^{\nu \times \nu}$ and $J(t, k) \in \mathbb{R}^{\nu \times \nu}$. This hybrid system flows and jumps according to some hybrid time domain to be specified. Note that generator~(\ref{equ:explicitGen}) is inherently a more direct representation of exogenous signals (\textit{i.e.} ``explicit'') when compared with other modeling frameworks. For example, a square wave, indicated by the symbol $\sqcap$, can be expressed by~(\ref{equ:explicitGen}) by setting $\Lambda\left(t, t_0\right) = \sqcap(t)$, directly, without the need of specifying whether this signal is generated by, \textit{e.g.} a nonlinear system $\sqcap(t)=\operatorname{sign}(\sin (t-t_0))$, or by the hybrid system~(\ref{equ:hybridGen}) for an opportune selection of constant matrices $S$ and $J$.

With this explicit generator~(\ref{equ:explicitGen}) at hand, we focus on the regulation problem with non-smooth exogenous signal $\omega$, and we discuss the assumptions that make the signal $\omega$ well-behaved. We first stress that, differently from the literature of output regulation of hybrid systems, we do not assume periodicity. With respect to non-smoothness, we consider piecewise continuous signals, as this does not exclude any signal of practical interest, \textit{e.g.} (possibly time-varying) square or triangular waves. 
We should also require some additional properties to bring the explicit generator~(\ref{equ:explicitGen}) closer to the standard output regulation setting. For instance, one would like the uniqueness of the solution $\omega$ generated by~(\ref{equ:explicitGen}) for all times. To guarantee this uniqueness, we require that $\Lambda$ is non-singular for all times. Also this assumption is not restrictive. For instance, any signal representable by (\ref{equ:implicitGen}) can be expressed by (\ref{equ:explicitGen}) with $\Lambda\left(t, t_0\right) = e^{S(t-t_0)}$, which is invertible for all times. This invertibility is also a standard property of the transition matrix of~(\ref{equ:timeVaryingGen}). More generally, given a signal of interest $\theta(t) = \tilde{P} \tilde{\omega}(t)$, the following proposition implies that it is always possible to construct $\Lambda$ invertible, possibly by inflating its dimension, that generates the same signal.
\begin{prop}\label{prop:LambdaDesign}
    Given a piecewise continuous signal $\theta(t) = \tilde{P} \tilde{\omega}(t)$, where $\tilde{P}\in\mathbb{R}^{1 \times \tilde{\nu}}$ and $\tilde{\omega}(t) \in \mathbb{R}^{\tilde{\nu}}$ satisfies $\tilde{\omega}(t) = [\Phi^{\omega}_1(t, t_0)\tilde{\omega}_1, \Phi^{\omega}_2(t, t_0)\tilde{\omega}_2$, $\cdots, \Phi^{\omega}_{\hat{\nu}}(t, t_0) \tilde{\omega}_{\hat{\nu}}]^{\top}$ with $\tilde{\omega}_i$'s arbitrary initial conditions and $i = 1, 2, \cdots, \tilde{\nu}$, it is always possible to construct a matrix-valued function $\Lambda \in \mathbb{R}^{\nu \times \nu}$ with $\nu \geq \tilde{\nu}$ and $\Lambda$ invertible for all times, a matrix  $\hat{P}\in\mathbb{R}^{1 \times \nu}$ and a vector $\omega_0\in\mathbb{R}^\nu$ such that $\theta(t) = \hat{P}\Lambda(t,t_0) \omega_0$.
\end{prop}
\begin{proof}
We start by proving that for each $\Phi^{\omega}_i$, there always exists an invertible $\Lambda_{i}(t, t_0) \in \mathbb{R}^{\nu_i \times \nu_i}$ with $\nu_i \in \{1, 2\}$ that can generate $\tilde{\omega}_{i}$. In fact, if $\Phi^{\omega}_i(t, t_0) \neq 0$ for all times, then $\Lambda_{i}(t, t_0) = \Phi^{\omega}_i(t, t_0)$ with $\Lambda_{i}$ invertible for all times. The less obvious case is $\Phi^{\omega}_i(t, t_0) = 0$ for some values of $t$. In this case, an invertible $\Lambda_{i}$ can be constructed as\footnote{ A practical example of this procedure is given in Section~\ref{sec:example}.}
\begin{equation*}
    \Lambda_{i}(t, t_0)=\left[\begin{array}{cr}
    \Phi^{\omega}_i(t, t_0) & -\hat{\Phi}^{\omega}_i(t, t_0) \\
    \hat{\Phi}^{\omega}_i(t, t_0) & \Phi^{\omega}_i(t, t_0)
    \end{array}\right],
\end{equation*}
with $\hat{\Phi}^{\omega}_i(t, t_0) \neq 0$ at least when $\Phi_i^\omega(t, t_0)=0$. Then, if $\hat{\nu} \geq 1$, $\Lambda$ can be constructed by $\Lambda = \operatorname{diag}(\Lambda_{1}, \Lambda_{2}, \cdots, \Lambda_{\hat{\nu}})$ and $\hat{P}$ and $\omega_0$ can be selected by inflating $\tilde{P}$ and $\tilde{\omega}_0$ with zeros.
\end{proof}

In addition, we need to consider the boundedness of $\Lambda$ and its inverse. The previous full-information study~\cite{ref:niu2024full} assumed that both $\Lambda$ and its inverse are bounded. In this article, we keep the boundedness of $\Lambda^{-1}$ to guarantee that the exogenous signal does not vanish to zero, but we relax the boundedness of $\Lambda$ to be just \textit{finite-time bounded}\footnote{This means that $\Lambda(t, t_0)$ is bounded for any finite-time $t$.}. This means that the exogenous signals can diverge to infinity, even faster than exponentially. This relaxed assumption is an extension of the traditional linear output regulation problem, in which the exosystem in implicit form~(\ref{equ:implicitGen}) is assumed to satisfy $\sigma(S) \in \mathbb{C}_{\geq 0}$~\cite[Assumption 1.1]{ref:huang2004nonlinear}. Note also that this assumption is less restrictive when compared with the standard hypothesis of the traditional nonlinear output regulation problem. In that case the exosystem $\dot{\omega} = s(\omega)$ is assumed to be Poisson stable~\cite[Section 8.1]{ref:isidori1995nonlinear}, \textit{i.e.}, all the eigenvalues of the matrix $[\frac{\partial s}{\partial \omega}]_{\omega = 0}$ lie on the imaginary axis and every trajectory of $\omega$ is periodic.


The aforementioned properties of the exogenous signal generator are formalized in the following assumption\footnote{Notably, we do not assume the semigroup property, namely we do not assume that $\Lambda(t_2,t_0)=\Lambda(t_2,t_1)\Lambda(t_1,t_0)$ for all $t_2 > t_1 > t_0$.}.
\begin{assumpt}\label{asmp:exoProp}
The matrix-valued function $\Lambda$ is piecewise continuous, non-singular, and finite-time bounded with $\Lambda^{-1}$ bounded for all times.
\end{assumpt}

Before looking into the regulation problem, we provide a preliminary technical result to characterise the asymptotic behaviour of a class of systems. This result plays an important role in solving the regulation problem in the later sections. To this end, we consider a linear time-varying system of the form
\begin{equation}\label{equ:expStableSystem}
    \dot{x}_g(t) = A_g(t)x_g(t) + B_g(t)\omega(t),
\end{equation}
where $x(t) \in \mathbb{R}^{g}$, $A_g(t) \in \mathbb{R}^{g \times g}$ and $B_g(t) \in \mathbb{R}^{g \times \nu}$ are bounded piecewise continuous, and $\omega$ is generated by the explicit exosystem~(\ref{equ:explicitGen}). Denote the state transition matrix of the system $\dot{x}_g(t) = A_g(t)x_g(t)$ as $\Phi_{g}(t, t_0)$, \textit{i.e.}, $\dot{\Phi}_{g}(t, t_0) = A_g(t)\Phi_{g}(t, t_0)$. Then the state trajectory of system~(\ref{equ:expStableSystem}) can be asymptotically characterized by the trajectories of $\omega$ under a certain stability condition, as shown next.

\begin{lem}\label{lem:moment}
Consider the interconnection of systems~(\ref{equ:expStableSystem}) and~(\ref{equ:explicitGen}). Suppose Assumption~\ref{asmp:exoProp} holds and system~(\ref{equ:expStableSystem}) with $\omega \equiv 0$ is exponentially stable. Then for any $\Pi_{g}(t_0) \in \mathbb{R}^{g \times \nu}$, the matrix-valued function $\Pi_{g}(t) \in \mathbb{R}^{g \times \nu}$ expressed by
\begin{equation}\label{equ:PIintegral}
\!\!\!\Pi_{g}(t) \!\!=\!\! \bigg(\!\! \Phi\!_{g}(t, t_0) \Pi_{g}(t_0) \!+\!\!\! \int_{t_0}^t \!\!\!\Phi_{g}(t, \tau) B_g(\tau) \Lambda(\tau, t_0) d \tau \!\! \bigg)\! \Lambda(t, t_0)\!^{-\!1}\!\!\!\!
\end{equation}
is bounded piecewise continuous and such that $\lim _{t \rightarrow+\infty}x_{g}(t)-\Pi_{g}(t) \omega(t)=\bm{0}_{g \times 1}$. Moreover, the matrix-valued function $\Psi_{g}(t) := \Pi_{g}(t)\Lambda(t, t_0) \in \mathbb{R}^{g \times \nu}$ is the unique solution to the differential equation
\begin{equation}\label{equ:PsiSystem}
    \dot{\Psi}_{g}(t) = A_g(t)\Psi_{g}(t) + B_g(t)\Lambda(t, t_0),
\end{equation}
and is such that $\lim _{t \rightarrow+\infty}x_{g}(t)-\Psi_{g}(t) \omega_0=\bm{0}_{n \times 1}$.
\end{lem}

\begin{proof}
See Appendix~\ref{equ:lemMomentProof} for a proof.
\end{proof}

\begin{remark}
    We highlight that the matrix-valued function $\Psi_{g}(t) := \Pi_{g}(t)\Lambda(t, t_0)$ is absolutely continuous (even when $\Lambda(t, t_0)$ is discontinuous) and the solution to the differential equation~(\ref{equ:PsiSystem}) 
    exists and is unique according to Carathéodory's theorem. This property, applied to analogous matrices, is used extensively in the paper.
\end{remark}

\begin{remark}
A result similar to that of Lemma~\ref{lem:moment} was first introduced in~\cite[Theorem 5.1]{ref:scarciotti2015model, ref:scarciotti2017nonlinear} in the context of model reduction by moment matching. However, that result required several additional technical assumptions (\textit{e.g.} polynomial boundedness) which have been dropped in the present formulation. As such, Lemma~\ref{lem:moment} is more general as it requires fewer and less restrictive assumptions. 
\end{remark}


\section{Full-Information Problem}\label{sec:FullInfoProb}
In this section, we study the full-information output regulation problem for linear systems with non-periodic non-smooth exogenous signals, and we consider a dynamic state feedback controller as the regulator to solve the problem. More specifically, we look for a dynamic state feedback controller of the form
\begin{equation}\label{equ:StateFdRegulator}
u(t) = K x(t) + \Gamma(t) \omega(t),
\end{equation}
with $K \in \mathbb{R}^{1 \times n}$, and $\Gamma(t) \in \mathbb{R}^{1 \times v}$ being a matrix-valued function to be designed. Note that differently from the standard linear theory of smooth regulation, $\Gamma$ is a time-varying matrix. The reason for this choice will be clear later. The problem is formulated as follows.

\begin{problem}[Full-information Non-smooth Output Regulation Problem]\label{prob:ORFullInfo}
Consider system~(\ref{equ:system}) interconnected with the exosystem~(\ref{equ:explicitGen}) under Assumption~\ref{asmp:exoProp}. The output regulation problem consists in designing a state-feedback regulator~(\ref{equ:StateFdRegulator}) such that:
\begin{itemize}[leftmargin=28pt]
    \item[\textbf{($\mathbf{S_F}$)}] The origin of the closed-loop system obtained by interconnecting system~(\ref{equ:system}), the exosystem (\ref{equ:explicitGen}), and the regulator with $\omega(t) \equiv 0$ is asymptotically stable.
    \item[\textbf{($\mathbf{R_F}$)}] The closed-loop system obtained by interconnecting system~(\ref{equ:system}), the exosystem~(\ref{equ:explicitGen}), and the regulator satisfies $\lim_{t \rightarrow +\infty} e(t) = 0$
    uniformly for any $x(t_0) \in \mathbb{R}^{n}$ and $\omega(t_0) \in \mathbb{R}^{\nu}$.
\end{itemize}
\end{problem}

To ensure that the problem can be solved, a standard assumption on system~(\ref{equ:system}) is introduced.
\begin{assumpt}\label{asmp:systemProp}
The pair $(A, \; B)$ is stabilizable.
\end{assumpt}

By interconnecting system~(\ref{equ:system}), exosystem~(\ref{equ:explicitGen}) and the state-feedback controller~(\ref{equ:StateFdRegulator}), we obtain the closed-loop system
\begin{equation}\label{equ:systemClosedLoopStateFd}
    \begin{aligned}
    \dot{x}(t) &= A_c x(t) + P_c(t) \omega(t),\\
    e(t) &= C_c x(t) + Q_c(t) \omega(t),\\
    \omega(t) &= \Lambda\left(t, t_0\right) \omega_0
    \end{aligned}
\end{equation}
with $A_c = A+B K$, $P_c(t) = P+B \Gamma(t)$, $C_c = C+D K$, and $Q_c(t) = Q+D \Gamma(t)$. Note that to solve the problem using regulator~(\ref{equ:StateFdRegulator}), we require that $\Gamma$ is piecewise continuous and bounded as $\omega$ is assumed to be piecewise continuous and finite-time bounded. The boundedness requirement of $\Gamma$ may not be clear at this stage, but this property is consistent with the traditional linear output regulation problem. In that case, there exists a constant gain matrix $\bar{\Gamma}$ solving the problem even when $\omega$ generated by~(\ref{equ:implicitGen}) is diverging, \textit{i.e.}, some of the eigenvalues of $S$ contain positive real part.

\subsection{Solution of the Full-Information Problem}\label{sec:solutionFIP}
Now we provide the solution to Problem~\ref{prob:ORFullInfo} in the form of regulator equations and we stress that the result of this section is instrumental for the development of solutions in the error-feedback case. 

\begin{thm}\label{thm:solutionWithD}
Consider Problem~\ref{prob:ORFullInfo}. Suppose Assumptions~\ref{asmp:exoProp} and~\ref{asmp:systemProp} hold. Then there exist a matrix $K$ and a bounded piecewise continuous $\Gamma(\cdot)$ such that the regulator~(\ref{equ:StateFdRegulator}) solves the full-information output regulation problem if and only if there exist bounded piecewise continuous matrix-valued functions $\Pi_x(t) \in \mathbb{R}^{n \times \nu}$ and $\Delta(t) \in \mathbb{R}^{1 \times \nu}$ that solve
\begin{equation}\label{equ:reguEquation}
    \begin{aligned}
    \Pi_x(t) \!\!&=\!\! \bigg(\!\! e^{A (t-t_0)} \Pi_x(t_0) \!+\!\!\! \int_{t_0}^t \!\!\!e^{A(t-\tau)}\! P_{\Delta}(\tau) \Lambda(\tau, t_0) d \tau \!\! \bigg)\! \Lambda(t, t_0)\!^{-\!1}\!\!\!, \\
    \bm{0}_{1 \times \nu} \!&= \lim_{t \rightarrow +\infty} C\Pi_x(t) + D\Delta(t) + Q.
    \end{aligned}
\end{equation}
for all $t \geq t_0$, where $P_{\Delta}(t) = P+B\Delta(t)$.
\end{thm}

\begin{proof}
First note that, under Assumption~\ref{asmp:systemProp}, there exists a matrix $K$ such that $\sigma(A+BK) \subset \mathbb{C}_{<0}$, \textit{i.e.}, condition \textbf{($\mathbf{S_F}$)} holds. Then we need to show that there exists a bounded piecewise continuous $\Gamma$ such that condition \textbf{($\mathbf{R_F}$)} holds if and only if there exist bounded piecewise continuous solutions $\Pi_x$ and $\Delta$ solving~(\ref{equ:reguEquation}). Now let $\Psi_x(t) \!=\! \Pi_x(t)\Lambda(t, t_0)$ and $\Gamma(t) = \Delta(t) - K\Pi_x(t)$. Then Lemma~\ref{lem:moment} implies that $\Psi_x$ is the unique solution to
\begin{equation*}
\dot{\Psi}_x(t) = A\Psi_x(t) + P_{\Delta}(t)\Lambda(t, t_0) = A_c\Psi_x(t) + P_c(t)\Lambda(t, t_0),
\end{equation*}
with $A_c = A + BK$ and $P_c(t) = P+B \Gamma(t)$. This result indicates that $\Psi_x$ can be equivalently expressed as
\begin{equation*}
    \Psi_x(t) = e^{A_c(t-t_0)}\Psi_x(t_0) + \int_{t_0}^t e^{A_c(t-\tau)} P_c(\tau) \Lambda(\tau, t_0) d \tau.
\end{equation*}
Therefore, by invertibility of $\Lambda$ and the selection $\Gamma = \Delta - K\Pi_x$, the equations (\ref{equ:reguEquation}) are equivalent to
\begin{equation}\label{equ:solveRWithD}
    \begin{aligned}
    \!\!\Pi_x\!(t) \!\!&=\!\! \bigg(\!\! e^{A_c (t-t_0)} \Pi_x(t_0) \!\!+\!\!\! \int_{t_0}^t \!\!\!e^{A_c\!(t-\tau)}\! P_c(\tau)\! \Lambda(\tau, t_0) d \tau \!\! \bigg)\! \Lambda(t, t_0)\!^{-\!1}\!\!\!\!, \\
    \!\!\bm{0}_{1 \times \nu} \!&= \lim_{t \rightarrow +\infty}(C + DK)\Pi_x(t) + Q + D\Gamma(t).
    \end{aligned}
\end{equation}

\textit{Sufficiency:} Suppose there exist bounded piecewise continuous $\Pi_x$ and $\Delta$ that solve (\ref{equ:reguEquation}). Then $\Pi_x$ and $\Gamma = \Delta - K\Pi_x$ also solve~(\ref{equ:solveRWithD}). By Lemma~\ref{lem:moment}, $\lim _{t \rightarrow+\infty}(x(t)-\Pi_x(t) \omega(t))=\bm{0}_{n \times 1}$. The regulation error of the closed-loop system~(\ref{equ:systemClosedLoopStateFd}) can be written as
\begin{equation}\label{equ:clErrorinPi}
    e(t) = C_c\bigl(x(t) - \Pi_x(t)\omega(t)\bigl) + \bigl(C_c\Pi_x(t) + Q_c(t)\bigl)\omega(t),
\end{equation}
where we recall that $C_c = C+D K$, and $Q_c(t) = Q+D \Gamma(t)$. Then when the second equation in (\ref{equ:solveRWithD}) holds, the regulation error given by~(\ref{equ:clErrorinPi}) is such that $\lim_{t \rightarrow +\infty} e(t) = 0$, \textit{i.e.}, condition \textbf{($\mathbf{R_F}$)} is satisfied.

\textit{Necessity:} Assume that there exist a matrix $K$ and a bounded piecewise continuous matrix $\Gamma$ satisfying conditions \textbf{($\mathbf{S_F}$)} and \textbf{($\mathbf{R_F}$)}. Then under condition \textbf{($\mathbf{S_F}$)}, Lemma~\ref{lem:moment} yields that for any bounded piecewise continuous $\Gamma$, there exists a bounded piecewise continuous $\Pi_x$ solving the first equation in~(\ref{equ:solveRWithD}) and such that $\lim _{t \rightarrow+\infty}x(t)-\Pi_x(t) \omega(t)=\bm{0}_{n \times 1}$. Since condition \textbf{($\mathbf{R_F}$)} also holds, by~(\ref{equ:clErrorinPi}), we have
\begin{equation}\label{equ:clErrorinPiLim}
    \begin{aligned}
    \lim_{t\rightarrow+\infty} \!\! e(t) 
    =& \!\!\lim_{t\rightarrow+\infty} \bigl((C \!+\! DK)\Pi_x(t) + (Q \!+\! D\Gamma(t)\bigl)\omega(t) = 0.
    \end{aligned}
\end{equation} 
By Assumption~\ref{asmp:exoProp}, $\Lambda(t, t_0)$ does not vanish as $t$ goes to infinity (because its inverse is bounded). Considering the arbitrarity of $\omega(t_0)$, this shows that for~(\ref{equ:clErrorinPiLim}) to hold, the second equation of~(\ref{equ:solveRWithD}) must be satisfied. Since~(\ref{equ:solveRWithD}) and~(\ref{equ:reguEquation}) are equivalent with $\Gamma = \Delta - K\Pi_x$, the selection $\Delta = K\Pi_x+\Gamma$ shows that there exist bounded and piecewise continuous $\Pi_x$ and $\Delta$ solving~(\ref{equ:reguEquation}).
\end{proof}


By Theorem~\ref{thm:solutionWithD}, if bounded and piecewise continuous matrices $\Pi_x$ and $\Delta$ solving~(\ref{equ:reguEquation}) can be found, then the control law
\begin{equation}\label{equ:controlLaw}
u(t) = K x(t) + (\Delta(t) - K\Pi_x(t)) \omega(t)
\end{equation}
solves Problem~\ref{prob:ORFullInfo} with $K$ being any matrix such that the closed-loop system is asymptotically stable, \textit{i.e.}, $\sigma(A+BK) \subset \mathbb{C}_{<0}$. Moreover, since Lemma~\ref{lem:moment} implies that $\lim _{t \rightarrow+\infty}x(t)-\Pi_x(t) \omega(t)=\bm{0}_{n \times 1}$, Theorem~\ref{thm:solutionWithD} and the control law~(\ref{equ:controlLaw}) indicate that there exists a trajectory $u^{*} = \Delta\omega$ such that any input $u$ in~(\ref{equ:StateFdRegulator}) that solves the output regulation problem satisfies $\lim_{t \rightarrow+\infty}u(t)-u^{*}(t)=0$. In addition, since $\Delta$ is bounded and piecewise continuous with $\omega$ finite-time bounded and piecewise continuous, the expected input $u^{*}$ that solves the problem is also finite-time bounded and piecewise continuous.

By mimicking the common terminology used in the classical output regulation problem~\cite{ref:huang2004nonlinear}, we refer to~(\ref{equ:reguEquation}) as \textit{regulator equations} in this non-smooth setting. However, differently from time-varying/hybrid output regulation problems~\cite{ref:zhang2006linear,ref:marconi2013internal}, the regulator equations~(\ref{equ:reguEquation}) are not differential-algebraic equations (DAE). Consequently, it is not straightforward to check the solvability conditions and compute solutions of the regulator equations. To remedy this problem, we introduce the following corollary that transforms~(\ref{equ:reguEquation}) into differential-algebraic equations.
\begin{corol}\label{corol:equivDAE}
Suppose Assumptions~\ref{asmp:exoProp} and~\ref{asmp:systemProp} hold. There exist bounded piecewise continuous matrices $\Pi_x$ and $\Delta$ that solve the regulator equations~(\ref{equ:reguEquation}) if and only if there exist a positive constant $\hat{t}$ and matrix-valued functions $\Psi_x(t) \in \mathbb{R}^{n \times \nu}$ and $\hat{\Delta}(t) \in \mathbb{R}^{1 \times \nu}$ solving
\begin{subequations}\label{equ:solutionDAE}
    \begin{align}
    \dot{\Psi}_x(t) &= A\Psi_x(t) + (B\hat{\Delta}(t) + P)\Lambda(t, t_0), \label{equ:solutionDAE_a}\\
    \bm{0}_{1 \times \nu} &= C \Psi_x(t) + (D\hat{\Delta}(t) + Q)\Lambda(t, t_0), \label{equ:solutionDAE_b}
    \end{align}
\end{subequations}
for all $t \geq \hat{t}$, where $\hat{\Delta}$ is bounded and piecewise continuous.
\end{corol} 
\begin{proof}
Under Assumption~\ref{asmp:exoProp}, $\Lambda$ does not decay to zero when $t \rightarrow +\infty$ as $\Lambda^{-1}$ is bounded. Then by right-multiplying $\Lambda(t)$ to~(\ref{equ:reguEquation}), the proof of Theorem~\ref{thm:solutionWithD} implies that there exist piecewise continuous $\Pi_x$ and $\Delta$ solving~(\ref{equ:reguEquation}) if and only if there exist $\Psi_x$ and $\hat{\Delta}$ (with $\hat{\Delta} = \Delta$) solving
\begin{equation*}
    \begin{aligned}
    \dot{\Psi}_x(t) &= A\Psi_x(t) + (B\hat{\Delta}(t) + P)\Lambda(t, t_0), \\
    \bm{0}_{1 \times \nu} &= \lim_{t \rightarrow +\infty} C \Psi_x(t) + (D\hat{\Delta}(t) + Q)\Lambda(t, t_0), \\
    \end{aligned}
\end{equation*}
which is equivalent to the existence of two initial conditions $\Psi_x(\hat{t})$ and $\hat{\Delta}(\hat{t})$ at a time instant $\hat{t} \geq t_0$ such that $\Psi_x(t)$ and $\hat{\Delta}(t)$ solve~(\ref{equ:solutionDAE}) for all $t \geq \hat{t}$.
\end{proof}

\begin{remark}
From a computational point of view, Corollary~\ref{corol:equivDAE} provides an alternative way of finding solutions to the regulator equations~(\ref{equ:reguEquation}). If one can find $\Psi_x$ and $\hat{\Delta}$ solving~(\ref{equ:solutionDAE}) for all $t \geq \hat{t}$ for some $\hat{t} \geq t_0$, then $\Pi_x = \Psi_x\Lambda^{-1}$ and $\Delta = \hat{\Delta}$ are the solutions to~(\ref{equ:reguEquation}). The simulation examples in this paper use this computational method. In addition, under certain assumptions (see \ifthenelse{\boolean{TAC}}{\cite[Appendix~C]{ref:Niu2025arXiv}}{ Appendix~\ref{sec:Appsolvability}}), the solvability of~(\ref{equ:solutionDAE}) is not related to the value of $\hat{t}$, \textit{i.e.}, we can set $\hat{t} = 0$ without loss of generality.
\end{remark}

\subsection{Solvability of the Regulator Equations}\label{sec:solvability}
As the solution of the output regulation problem depends on
solving the regulator equations~(\ref{equ:reguEquation}), it is necessary to study under which conditions such equations are solvable for any matrices $P$ and $Q$. Solvability of~(\ref{equ:reguEquation}) is technically involved. \ifthenelse{\boolean{TAC}}{Because of space constraints, and to}{To}
avoid distracting the reader from the main results of the paper (on output regulation), we have provided here a discussion and the main results, while we have moved the auxiliary technical derivation in \ifthenelse{\boolean{TAC}}{\cite[Appendix~C]{ref:Niu2025arXiv}.}{Appendix~\ref{sec:Appsolvability}.}

Given the generality of $\Lambda$ considered so far, the presence of the feedforward term $D$ plays a pivotal role in the tracking/rejection ability of the linear system~(\ref{equ:system}). This is formalised in the next result.


\begin{thm}\label{thm:zeroError}
Consider the interconnection of systems (\ref{equ:system}) and (\ref{equ:explicitGen}) with $D = 0$. Suppose Assumptions~\ref{asmp:exoProp} and~\ref{asmp:systemProp} hold. Then there exists a finite-time bounded piecewise continuous control input $u$ of the form~(\ref{equ:StateFdRegulator}) such that conditions \textbf{($\mathbf{S_F}$)} and \textbf{($\mathbf{R_F}$)} hold only if there exists a positive constant $\hat{t}$ such that the function $Q\omega(\cdot)$ is locally Lipschitz continuous for all $t \geq \hat{t}$. 
\end{thm}

\begin{proof}
See Appendix~\ref{AppB} for a proof.
\end{proof}

The significance of Theorem~\ref{thm:zeroError} is that the posed regulation problem cannot be solved (without feedforward term and with a finite-time bounded regulator\footnote{Note that if $u$ is allowed to be an impulsive control, then the problem with discontinuous signals may be solved without feedforward term. This is left as a future research direction.}) for exogenous signals with $Q\omega$ presenting discontinuities (\textit{e.g.} possibly time-varying square waves) even as $t \rightarrow +\infty$. Note that this makes sense: a discontinuous signal cannot be tracked by $Cx$, which is a continuous signal, without the help of a feedforward term $Du$ that can cancel the jump. In this section, we mainly focus on the case $D = 0$, as the case $D\ne 0$ is much simpler. We introduce the following assumption that helps us to guarantee the existence of a solution to the regulation problem when $D = 0$.

\begin{assumpt}\label{asmp:pwDiff}
The matrix-valued function $Q\Lambda(t, t_0)$ is piecewise differentiable\footnote{By definition, this means that for any $t_b > t_a \geq t_0$, there exists a finite subdivision $t_a = \hat{t}_0 < \hat{t}_1 < \cdots <\hat{t}_{k-1} < \hat{t}_{k} = t_b$ of $[t_a,\;t_b]$ such that $Q\Lambda(t, t_0)$ is continuously differentiable in each subinterval $[t_{i-1},\; t_{i}]$ for any $i = 1, 2, \cdots, k$. Note that the derivative at $t_{i-1}$ is understood as the right derivative and the derivative at $t_{i}$ is understood as the left derivative~\cite[Definition 3.1]{ref:do1992riemannian}.} and $Q\dot{\Lambda}(t, t_0)\Lambda^{-1}(t, t_0)$ is bounded for all $t \in \mathcal{T}$, where $\mathcal{T}$ denotes any time interval in which $Q\Lambda(t, t_0)$ is continuously differentiable.
\end{assumpt}

Assumption~\ref{asmp:pwDiff} is a stronger assumption compared with the local Lipschitz continuity requirement on $Q\omega$. On the one hand, the piecewise differentiability assumption implies that $Q\omega$ has a finite number of non-differentiable points in any finite interval and is semi-differentiable at each non-differentiable point, \textit{i.e.}, both left and right derivatives exist, although different. This differentiability requirement does not practically restrict the applicability of the result when considering exogenous signals such as (possibly time-varying) triangular waveforms. On the other hand, the boundedness of $Q\dot{\Lambda}(t, t_0)\Lambda^{-1}(t, t_0)$ can be seen as a direct extension of the continuous-time output regulation problem to the non-smooth case. More specifically, if the exosystem is as in~(\ref{equ:timeVaryingGen}) with $\Lambda$ the corresponding state-transition matrix, then the boundedness of $Q\dot{\Lambda}(t, t_0)\Lambda^{-1}(t, t_0) = Q\tilde{S}(t)$ naturally holds as $\tilde{S}(t)$ is classically assumed to be bounded for all times in the time-varying setting~\cite{ref:zhang2009adaptive}. It should be pointed out that this assumption on $Q\omega$ does not imply that $\omega$ cannot be discontinuous. In fact, the solvability can still be satisfied when the jumps of $\omega$ occur in the kernel of the matrix $Q$ for all times. In other words, we do not exclude discontinuous signals from $P\omega$.

We also note that when $D = 0$, the solvability of the problem is closely related to the relative degree\footnote{The smallest positive integer $r \leq n$ such that $C A^{r-1} B \neq 0$.} of system~(\ref{equ:system}). While this may seem surprising (because it is well known that the relative degree does not play a role in the solvability of the classical regulator equations), in our problem setting, the relative degree controls the ``degree of non-smoothness'' of the state $x$ (this concept is made precise in the analysis in \ifthenelse{\boolean{TAC}}{\cite[Appendix~C]{ref:Niu2025arXiv}.}{Appendix~\ref{sec:Appsolvability}.}). In other words, the relative degree determines the ability of the system~(\ref{equ:system}) with $\omega \equiv 0$ and a locally bounded input $u$ to produce a (possibly non-smooth) output that can cancel the (possibly non-smooth) output of system~(\ref{equ:system}) with $u \equiv 0$. Therefore, when $D = 0$, we study the solvability of the regulator equations~(\ref{equ:reguEquation}) with an eye to the relative degree of system~(\ref{equ:system}). 
To this end, we introduce the following assumption.
\begin{assumpt}\label{asmp:unitaryRD}
System~(\ref{equ:system}) has a unitary relative degree.
\end{assumpt}

\noindent At this point we introduce a new non-resonance condition.




\begin{define}[Non-smooth Non-resonance Condition]
Systems~(\ref{equ:system}) and~(\ref{equ:explicitGen}) are \textit{non-resonant} if the matrix-valued function
\begin{equation}\label{equ:nonSmNRC}
    \Omega(t)= \int_{t_0}^t \left(
    \Lambda(\tau, t_0)\Lambda(t, t_0)^{-1}\right)^{\top} \otimes e^{A_{z}(t - \tau)} d \tau,
\end{equation}
is bounded for all $t \geq t_0$, where $A_{z} \in \mathbb{R}^{(n-1) \times (n-1)}$ has eigenvalues identical to the transmission zeros of system~(\ref{equ:system}).
 \end{define}

\begin{remark}\label{remark:minphase}
By Lemma~\ref{lem:moment} a sufficient condition for the non-resonance condition to hold under Assumption~\ref{asmp:exoProp} is that the linear system~(\ref{equ:system}) is minimum phase. In fact, by vectorizing (\ref{eq-limPI_g}) in the proof of Lemma~\ref{lem:moment} yields (\ref{equ:nonSmNRC}) times $\text{vec}\left(B_g(\tau)\right)$, with $\Phi_g(t,\tau) = e^{A_z(t-\tau)}$. If system~(\ref{equ:system}) is minimum phase, $\Phi_g(t,\tau)$ is exponentially stable. Then the boundedness of $\Omega$ in~(\ref{equ:nonSmNRC}) follows from (\ref{eq-limPI_g}).
\end{remark}

Leveraging the non-resonance condition, it is possible to obtain the following result.

\begin{thm}\label{thm:solCondBoth}
Consider Problem~\ref{prob:ORFullInfo} driven by control law~(\ref{equ:controlLaw}). Suppose Assumptions~\ref{asmp:exoProp} and~\ref{asmp:systemProp} hold. There exist bounded piecewise continuous matrices $\Pi_x$ and $\Delta$ solving the regulator equations~(\ref{equ:reguEquation}) for any $P$ and $Q$ if and only if systems~(\ref{equ:system}) and~(\ref{equ:explicitGen}) are non-resonant and either of the following holds.
\begin{itemize}
\item $D = 0$, and Assumptions~\ref{asmp:pwDiff} and \ref{asmp:unitaryRD} hold.
\item $D \ne 0$.
\end{itemize}
\end{thm}

The proof of Theorem~\ref{thm:solCondBoth} is involved and requires several additional technical results. As already mentioned, we have placed the proof in \ifthenelse{\boolean{TAC}}{\cite[Appendix~C]{ref:Niu2025arXiv}.}{Appendix~\ref{sec:Appsolvability}.}

\begin{remark}
The non-smooth non-resonance condition is an extension of the non-resonance condition in the traditional linear output regulation problem~\cite[Assumption 1.4]{ref:huang2004nonlinear} to the non-smooth case. To be more specific, if $\Lambda(t, t_0) = e^{S(t-t_0)}$ with $S$ in~(\ref{equ:implicitGen}), then the boundedness of $\Omega$ in~(\ref{equ:nonSmNRC}) implies (but is not implied by) the non-singularity of $I \otimes A_{z} -S^{\top} \otimes I$. This is equivalent to $\sigma(A_{z}) \cap \sigma(S) = \varnothing$, which is exactly the non-resonance condition defined in the traditional linear output regulation problem, as the eigenvalues of $A_{z}$ coincides with the transmission zeros of system~(\ref{equ:system}). On the other hand, the boundedness of $\Omega$ in~(\ref{equ:nonSmNRC}) is a stronger condition when compared with the traditional non-resonance condition. This is due to the extension of the Sylvester equation to the non-periodic non-smooth case. Additional insights are provided in \ifthenelse{\boolean{TAC}}{\cite[Appendix~C]{ref:Niu2025arXiv}.}{Appendix~\ref{sec:Appsolvability}.}
\end{remark}

\section{Error-Feedback Problem}\label{sec:ErrorFdProb}
In this section, we study the non-smooth output regulation problem when the states of the system are not known, but the regulation error $e(t)$ is assumed to be available to the regulator. In this context, we look at the construction of an error-feedback regulator of the form
\begin{equation}\label{equ:errorFdRegulator}
\begin{aligned}
    \dot{\xi}(t) &= \mathcal{G}_\xi(t)\xi(t) + \mathcal{G}_e(t) e(t), \\
    u(t) &= K_\xi(t) \xi(t) + K_e(t) e(t).
\end{aligned}
\end{equation}
with $\xi(t) \in \mathbb{R}^q$, and bounded piecewise continuous matrix-valued functions $\mathcal{G}_\xi(t) \in \mathbb{R}^{q \times q}$, $\mathcal{G}_e(t) \in \mathbb{R}^{q \times 1}$, $K_\xi(t) \in \mathbb{R}^{1 \times q}$, and $K_e(t) \in \mathbb{R}$. Then the error-feedback problem is formulated as follows.

\begin{problem}[Error-feedback Non-smooth Output Regulation Problem]\label{prob:ORErrorFd}
Consider system~(\ref{equ:system}) interconnected with the exosystem~(\ref{equ:explicitGen}) under Assumption~\ref{asmp:exoProp}. The output regulation problem consists in designing a regulator~(\ref{equ:errorFdRegulator}) such that:
\begin{itemize}[leftmargin=28pt]
    \item[\textbf{($\mathbf{S_E}$)}] The origin of the closed-loop system obtained by interconnecting system~(\ref{equ:system}), the exosystem (\ref{equ:explicitGen}), and the regulator with $\omega(t) \equiv 0$ is exponentially stable.
    \item[\textbf{($\mathbf{R_E}$)}] The closed-loop system obtained by interconnecting system~(\ref{equ:system}), the exosystem~(\ref{equ:explicitGen}), and the regulator satisfies 
    $\lim_{t \rightarrow +\infty} e(t) = 0,$
    uniformly for any $\left(w(t_0), x(t_0), \xi(t_0)\right) \in \mathbb{R}^{\nu+n+q}$.
\end{itemize}
\end{problem}

Note that, due to the unbounded nature of the exogenous signal $\omega$, Problem~\ref{prob:ORErrorFd} requires, in ($\mathbf{S_E}$), the exponential stability of the closed-loop system rather than simply uniform asymptotic stability which is required in the time-varying output regulation problems~\cite{ref:zhang2006linear,ref:zhang2009adaptive}. It will be clear later that this exponential stability is crucial for guaranteeing bounded solutions in the error-feedback setting 
and can be easily achieved by a high-gain stabilisation method.

To study this problem, we first characterise the ``internal model property'' of the regulator, and then we propose solutions to Problem~\ref{prob:ORErrorFd} by designing a regulator who possesses such property. We mainly focus our attention on the case with zero feedforward matrix, (\textit{i.e.}, $D = 0$) as the systems with a non-zero relative degree are more commonly encountered in real-life applications such as electrical engineering, aeronautics, and mechanical manipulations~\cite{ref:juang2001identification}. At the end of this section, we will discuss the solution to the problem when the relative degree is zero ($D \neq 0$).


\subsection{Internal Model Property}
\label{sec:IMProp}

Denoting $\zeta = \operatorname{col}(x, \xi)$, the closed-loop system~(\ref{equ:explicitGen})--(\ref{equ:system})--(\ref{equ:errorFdRegulator}) can be written as
\begin{equation}\label{equ:systemClosedLoopErrorFd}
    \begin{aligned}
    \dot{\zeta}(t) &= \mathcal{A}_{c}(t)\zeta(t) + \mathcal{P}_{c}(t) \omega(t),\\
    e(t) &= C x(t) + Q \omega(t), \\
    \omega(t) &= \Lambda\left(t, t_0\right) \omega_0,
    \end{aligned}
\end{equation}
with
\begin{equation*}
\begin{aligned}
\mathcal{A}_{c}(t)\!\!=\!\!\left[\!\!\!\begin{array}{cc}
A+B K_e(t) C & B K_\xi(t) \\
\mathcal{G}_e(t) C & \mathcal{G}_\xi(t)
\end{array}\!\!\!\right]\!\!, 
& \,\,\,\mathcal{P}_{c}(t)\!\!=\!\!\left[\!\!\!\begin{array}{c}
P + B K_e(t) Q \\
\mathcal{G}_e(t) Q
\end{array}\!\!\!\right]\!\!. 
\end{aligned}
\end{equation*}
For simplicity, we let $\Phi_{c}(t, t_0)$ be the state transition matrix of the system $\dot{\zeta}(t) = \mathcal{A}_{c}(t)\zeta(t)$, \textit{i.e.}, $\dot{\Phi}_{c}(t, t_0) = \mathcal{A}_{c}(t)\Phi_{c}(t, t_0)$. Then we show that the error-feedback regulator~(\ref{equ:errorFdRegulator}) solving Problem~\ref{prob:ORErrorFd} must possess an internal model property.

\begin{thm}\label{thm:IMProperty}
Consider Problem~\ref{prob:ORErrorFd} and let Assumptions~\ref{asmp:exoProp} and~\ref{asmp:pwDiff} hold. Suppose there exists an error-feedback controller~(\ref{equ:errorFdRegulator}) such that the condition \textbf{($\mathbf{S_E}$)} is satisfied. Then such a controller satisfies condition \textbf{($\mathbf{R_E}$)} if and only if there exist a positive constant $\hat{t}$ and bounded piecewise continuous $\Pi_x(t) \in \mathbb{R}^{n \times \nu}$, $\Pi_\xi(t) \in \mathbb{R}^{q \times \nu}$, and $\Delta(t) \in \mathbb{R}^{1 \times \nu}$ solving the regulator equations~(\ref{equ:reguEquation}) and the equation
\begin{subequations}\label{equ:IMProperty}
\begin{align}
    \frac{d}{dt}\left(\Pi_\xi(t)\Lambda(t, t_0) \right) &= \mathcal{G}_\xi(t)\Pi_\xi(t)\Lambda(t, t_0), \label{equ:IMProperty_a}\\
    \Delta(t) &= K_\xi(t)\Pi_\xi(t), \label{equ:IMProperty_b}
\end{align}
\end{subequations}
for all $t \geq \hat{t}$.
\end{thm}
\begin{proof}
Consider the closed-loop system~(\ref{equ:systemClosedLoopErrorFd}) with a controller~(\ref{equ:errorFdRegulator}) such that  condition \textbf{($\mathbf{S_E}$)} is satisfied. Then Lemma~\ref{lem:moment} yields that there exists a bounded piecewise continuous $\Pi_\Phi(t) \in \mathbb{R}^{(n+q) \times \nu}$ in
\begin{equation*}
\!\!\Pi_\Phi\!(t) \!\!=\!\! \bigg(\!\! \Phi_{c}\!(t, t_0) \Pi_\Phi\!(t_0) \!+\!\!\! \int_{t_0}^t \!\!\!\!\Phi_{c}(t, \tau) \mathcal{P}_c(\tau) \Lambda(\tau, t_0) d \tau \!\! \bigg)\! \Lambda(t, t_0)\!^{-\!1}\!\!\!\!
\end{equation*}
such that $\lim_{t \rightarrow +\infty} \zeta(t) - \Pi_\Phi(t) \omega(t) = 0$, or equivalently, $\lim_{t\rightarrow+\infty} \zeta(t) - \Psi_\Phi(t)\omega_0 = 0$ with $\Psi_\Phi(t) = \Pi_\Phi(t)\Lambda(t, t_0)$ being the solution to the differential equation
\begin{equation*}
    \dot{\Psi}_\Phi(t) 
    = \mathcal{A}_c(t)\Psi_\Phi(t) + \mathcal{P}_c(t)\Lambda(t, t_0).
\end{equation*}
Let $\Psi_\Phi = [\Psi_x^{\top}, \Psi_\xi^{\top}]^{\top}$, with $\Psi_x(t) \in \mathbb{R}^{n \times \nu}$ and $\Psi_\xi(t) \in \mathbb{R}^{q \times \nu}$, and $\Pi_\Phi = [\Pi_x^{\top}, \Pi_\xi^{\top}]^{\top}$. By denoting $\Delta(t) = K_\xi(t)\Pi_\xi(t) = K_\xi(t) \Psi_\xi(t)\Lambda(t, t_0)^{-1}$, the last differential equation is equivalent to
\begin{subequations}\label{equ:IMPropertyPsi}
\begin{align}
    \dot{\Psi}_x(t) =& \;A\Psi_x(t) + B\Delta(t)\Lambda(t, t_0) + P\Lambda(t, t_0)\nonumber\\ &+ BK_e(t)C \Psi_x(t) + BK_e(t)Q\Lambda(t, t_0), \label{equ:IMPropertyPsi_a}\\
    \dot{\Psi}_\xi(t) =& \;\mathcal{G}_\xi(t)\Psi_\xi(t) + \mathcal{G}_e(t)(C \Psi_x(t) + Q\Lambda(t, t_0)), \label{equ:IMPropertyPsi_b}\\
    \Delta(t) =& \;K_\xi(t)\Psi_\xi(t)\Lambda(t, t_0)^{-1},\label{equ:IMPropertyPsi_c}
\end{align}
\end{subequations}
for all $t \geq \hat{t}$. By using the transformation
\begin{equation*}
\tilde{\zeta}(t):= \left[\begin{array}{c}
\tilde{x}(t) \\
\tilde{\xi}(t)
\end{array}\right] = \left[\begin{array}{c}
x(t)-\Psi_x(t) \omega_0 \\
\xi(t)-\Psi_\xi(t) \omega_0
\end{array}\right]
\end{equation*}
on the first and third equation in (\ref{equ:systemClosedLoopErrorFd}) yields the system
\begin{subequations}
\begin{align}  
    \dot{\tilde{\zeta}}(t) &= \mathcal{A}_c(t)\tilde{\zeta}(t), \label{equ:closeLoopErrorFdStable} \\
    \omega(t) &= \Lambda(t, t_0) \omega_0,
\end{align}
\end{subequations}
for all $t \geq t_0$ with an arbitrary initial condition $\tilde{\zeta}_0 = \tilde{\zeta}(t_0)$. Accordingly, the regulation error $e$ in (\ref{equ:systemClosedLoopErrorFd}) can be expressed as
\begin{equation}\label{equ:errorErrorFd}
    e(t) = C \tilde{x}(t) + (C\Psi_x(t) + Q\Lambda(t, t_0))\omega_0.
\end{equation}
Then, by the exponential stability of~(\ref{equ:closeLoopErrorFdStable}), condition \textbf{($\mathbf{R_E}$)} (\textit{i.e.}, $\lim_{t\rightarrow+\infty} e(t) = 0$) is satisfied for any arbitrary $\omega_0$ if and only if there exists a positive constant $\hat{t}$ such that
\begin{equation}\label{equ:algebriacPsi}
    C\Psi_x(t) + Q\Lambda(t, t_0) = 0
\end{equation}
for all $t \geq \hat{t}$. Finally,~(\ref{equ:IMPropertyPsi_a}), ~(\ref{equ:IMPropertyPsi_b}), and~(\ref{equ:algebriacPsi}) are equivalent to (\ref{equ:solutionDAE_a}), (\ref{equ:IMProperty_a}) (with $\Pi_\xi(t) = \Psi_\xi(t)\Lambda(t, t_0)^{-1}$), and (\ref{equ:solutionDAE_b}) (with $D = 0$), respectively. By Corollary~\ref{corol:equivDAE},~(\ref{equ:solutionDAE}) is equivalent to (\ref{equ:reguEquation}) with $\Pi_x(t) = \Psi_x(t)\Lambda(t, t_0)^{-1}$, and $\hat{\Delta}=\Delta$ given in (\ref{equ:IMProperty_b}). This concludes the proof. 
\end{proof}

By analogy with the terminology used in the time-invariant setting (see~\cite{ref:francis1976internal, ref:byrnes2003limit}), we say that an error-feedback regulator~(\ref{equ:errorFdRegulator}) that satisfies~(\ref{equ:IMProperty}) possesses an internal model property, as defined next.
\begin{define}[Explicit Internal Model Property]
A regulator that has the ability to generate all possible signals produced by the output $y_\omega$ of the system
\begin{equation}\label{equ:IMsystem}
\begin{aligned}  
    \omega(t) &= \Lambda(t, t_0) \omega_0, \\
    y_\omega(t) &= \Delta(t)\omega(t)
\end{aligned}
\end{equation}
when the regulation error $e$ is identically zero, is said to have the \textit{explicit internal model property}.
\end{define}

The explicit internal model property 
is guaranteed by the existence of a bounded piecewise continuous $\Pi_\xi$ solving~(\ref{equ:IMProperty}). In fact, by the proof of Theorem~\ref{thm:IMProperty},  the state $\xi$ of a regulator~(\ref{equ:errorFdRegulator}) that exponentially stabilises the closed-loop system and for which there exists a bounded continuous $\Pi_\xi$ solving~(\ref{equ:IMProperty}), has the property that 
$\lim_{t \rightarrow +\infty} \xi - \Pi_\xi\omega = 0$ while $\lim_{t \rightarrow +\infty} e = 0$. Then~(\ref{equ:IMProperty}) requires that the output of regulator~(\ref{equ:errorFdRegulator}) satisfies $\lim_{t \rightarrow +\infty} u = \lim_{t \rightarrow +\infty} K_\xi \xi + K_e e = \lim_{t \rightarrow +\infty} K_\xi \Pi_{\xi} \omega = \Delta \omega$, \textit{i.e.},  a regulator has the explicit internal model property.

\begin{remark}\label{mrk:IMPvsTimeVarying}
Consider the time-varying linear output regulation problem with periodic smooth exogenous signals generated by a time-varying implicit generator of the form~(\ref{equ:timeVaryingGen}), see~\cite{ref:zhang2006linear,ref:zhang2009adaptive}. In that framework, the internal model property is characterized by DAEs composed of a time-varying Sylvester equation and an algebraic equation that coincides with the second equation of~(\ref{equ:IMProperty}). In our framework, the first equation of~(\ref{equ:IMProperty}), when $\Lambda$ is differentiable, can be shown to be equivalent to the time-varying Sylvester equation
\begin{equation*}
    \dot{\Pi}_{\xi}(t) = \mathcal{G}_{\xi}(t)\Pi_\xi(t) - \Pi_{\xi}(t)\dot{\Lambda}(t, t_0)\Lambda(t, t_0)^{-1}.
\end{equation*}
In particular, if $\Lambda$ is the transition matrix of generator~(\ref{equ:timeVaryingGen}), this differential equation coincides with the time-varying Sylvester equation that characterizes the internal model property in the smooth setting~\cite[equation~(10)]{ref:zhang2009adaptive} (noting that $\dot \Lambda (t,t_0) \Lambda(t,t_0)^{-1} = \tilde{S}(t)$).
\end{remark}

\subsection{Achieving Internal Model Property}
\label{sec:AchievingIMProp}
Now, with the solution $\Delta$ of the regulator equations, we study how to design the matrices $\mathcal{G}_\xi(t)$ and $K_\xi(t)$ of the controller~(\ref{equ:errorFdRegulator})
to secure the internal model property. In the smooth output regulation problems, this is normally done following the idea called \textit{canonical realization} proposed in the nonlinear case in~\cite{ref:serrani2001semi} and reformulated in the smooth periodic output regulation problem in~\cite[Definition 3.2]{ref:zhang2009adaptive}. In those problems, they considered a model 
\begin{subequations}\label{equ:IMsystemSmooth}
    \begin{align}
    \dot{\tilde{\eta}}(t) &= \tilde{\Phi}_{\mathrm{im}}(t)\tilde{\eta}(t),
    \label{equ:IMsystemSmooth1}  \\
    \tilde{v}(t) &= \tilde{\Xi}_{\mathrm{im}}(t)\tilde{\eta}(t),
    \end{align}
\end{subequations}
with bounded $\tilde{\Phi}_{\mathrm{im}}: \mathbb{R} \to \mathbb{R}^{s_1 \times s_1}$ and $\tilde{\Xi}_{\mathrm{im}}: \mathbb{R} \to \mathbb{R}^{1 \times s_1}$ having the internal model property. Then, they designed a canonical realization of (\ref{equ:IMsystemSmooth}) by immersing it into a linear time-varying system $\left(F_{\mathrm{im}}^{*}(\cdot), G_{\mathrm{im}}^{*}(\cdot), H_{\mathrm{im}}^{*}(\cdot)\right)$ of the form
\begin{equation}\label{equ:canonicalImmersion}
\begin{aligned}
    \dot{\eta}^{*}(t) &= (F_{\mathrm{im}}^{*}(t) + G_{\mathrm{im}}^{*}(t) H_{\mathrm{im}}^{*}(t))\eta^{*}(t), \\
    v^{*}(t) &= H_{\mathrm{im}}^{*}(t) \eta^{*}(t),
\end{aligned}
\end{equation}
as this form enables the stabilization problem to be solved in a relatively simple manner (this will be clear later). In the non-smooth case, we can pursue a similar strategy by considering a model in explicit form
\begin{subequations}\label{equ:IMsystemNonsmooth}
\begin{align}  
    \hat{\eta}(t) &= \hat{\Phi}_{\mathrm{im}}(t, t_0) \hat{\eta}_0,  \label{equ:IMsystemNonsmooth1} \\
    \hat{v}(t) &= \hat{\Xi}_{\mathrm{im}}(t)\hat{\eta}(t),
\end{align}
\end{subequations}
with piecewise continuous $\hat{\Phi}_{\mathrm{im}}: \mathbb{R} \times \mathbb{R} \to \mathbb{R}^{s_2 \times s_2}$ and $\hat{\Xi}_{\mathrm{im}}: \mathbb{R} \to \mathbb{R}^{1 \times s_2}$, where $\hat{\Phi}_{\mathrm{im}}$ has the same properties as $\Lambda$ under Assumption~\ref{asmp:exoProp}, having the explicit internal model property. A trivial example for such a model would be $\hat{\Phi}_{\mathrm{im}} = \Lambda$ and $\hat{\Xi}_{\mathrm{im}} = \Delta$.
Since the regulator~(\ref{equ:errorFdRegulator}) is in differential form, even in this possibly non-smooth case, we want to construct an immersion of (\ref{equ:IMsystemNonsmooth})
into a system of the form~(\ref{equ:canonicalImmersion}). To this end, and because it will be used also later, we first review the definition of canonical realisation for the implicit internal model~(\ref{equ:IMsystemSmooth}). Then we generalise the definition to the explicit case.

\begin{define}[Canonical Realization]\label{def:canonicalReal}
Consider an internal model of the form~(\ref{equ:IMsystemSmooth}). We say that the pair $(\tilde{\Phi}_{\mathrm{im}}(\cdot),\; \tilde{\Xi}_{\mathrm{im}}(\cdot))$ admits a \textit{canonical realization} if there exist a bounded continuous mapping $M: \mathbb{R} \to \mathbb{R}^{m \times s_1}$, with $m \geq s_1$, and a time-varying triple $\left(F_{\mathrm{im}}(\cdot), G_{\mathrm{im}}(\cdot), H_{\mathrm{im}}(\cdot)\right) \in \mathbb{R}^{m \times m} \times \mathbb{R}^{1 \times m} \times \mathbb{R}^{m \times 1}$ such that 
\begin{itemize}
    \item[i)] The system
    \begin{equation}\label{equ:FimSystem}
        \dot{x}_f(t) = F_{\mathrm{im}}(t)x_f(t)
    \end{equation}
    is exponentially stable.
    \item[ii)] $M$ has constant positive rank $p_1 \leq s_1$, and satisfies
    \begin{equation}\label{equ:canonicalRealSylv}
        \begin{aligned}
        \dot{M}(t) \!+\! M(t) \tilde{\Phi}_{\mathrm{im}}(t) &= \left(F_{\mathrm{im}}(t) + G_{\mathrm{im}}(t) H_{\mathrm{im}}(t)\right)\! M(t), \\
        \tilde{\Xi}_{\mathrm{im}}(t) & =H_{\mathrm{im}}(t) M(t),
        \end{aligned}
    \end{equation}
    for all $t \geq \hat{t}$ with some positive constant $\hat{t}$. 
\end{itemize}
\end{define}
\begin{define}[Explicit Canonical Realization]\label{def:NonsmoothcanonicalReal}
Consider an internal model of the form~(\ref{equ:IMsystemNonsmooth}), we say that the pair $(\hat{\Phi}_{\mathrm{im}}(\cdot),\; \hat{\Xi}_{\mathrm{im}}(\cdot))$ admits a \textit{canonical realization} if there exist a bounded piecewise continuous mapping $\Pi_M: \mathbb{R} \to \mathbb{R}^{m \times s_2}$, with $m \geq s_2$, and a time-varying triple $\left(F_{\mathrm{im}}(\cdot), G_{\mathrm{im}}(\cdot), H_{\mathrm{im}}(\cdot)\right)$ such that
\begin{itemize}
    \item[i)] System~(\ref{equ:FimSystem}) is exponentially stable.
    \item[ii)] $\Pi_M$ has constant positive rank $p_2 \leq s_2$, and satisfies
    \begin{subequations}\label{equ:IMCanoniParamMoment}
        \begin{align}
        \Pi_M(t) &= \bigg(\Phi_F(t,\hat{t})\Pi_M(\hat{t}) \notag \\
        &\!\!\!\!\!\!\!\!\!\!\!\!
        +\!\! \int_{\hat{t}}^t \!\! \Phi_F(t,\tau)G_\Xi(\tau) \hat{\Phi}_\mathrm{im}(\tau, t_0) d \tau \!\bigg)\! \hat{\Phi}_\mathrm{im}(t, t_0)^{-1}\!\!, \label{equ:IMCanoniParamPi} \\
        \hat{\Xi}_\mathrm{im}(t) &= H_{\mathrm{im}}(t)\Pi_M(t), \label{equ:IMCanoniParamPiAlgeb}
        \end{align}
    \end{subequations}
    for all $t \geq \hat{t}$ with some positive constant $\hat{t}$, where $G_\Xi(t) = G_{\mathrm{im}} (t)\hat{\Xi}_\mathrm{im}(t)$ and $\Phi_F(t, t_0) \in \mathbb{R}^{m \times m}$ is the state-transition matrix of system~(\ref{equ:FimSystem}).  
\end{itemize}
\end{define}

\begin{remark}
Similarly to the discussion in Remark~\ref{mrk:IMPvsTimeVarying}, the explicit expression~(\ref{equ:IMCanoniParamPi}) can be seen as a non-smooth version of the time-varying Sylvester equation shown in~(\ref{equ:canonicalRealSylv}). In fact, if systems~(\ref{equ:IMsystemSmooth1}) and~(\ref{equ:IMsystemNonsmooth1}) are equivalent, \textit{i.e.}, $\hat{\Phi}_{\mathrm{im}}$ in~(\ref{equ:IMsystemNonsmooth1}) is the transition matrix of system~(\ref{equ:IMsystemSmooth1}), then~(\ref{equ:canonicalRealSylv}) and~(\ref{equ:IMCanoniParamMoment}) are equivalent with $M = \Pi_M$.
\end{remark}

Definitions~\ref{def:canonicalReal} and~\ref{def:NonsmoothcanonicalReal} propose canonical realisations of a valid internal model in the implicit/explicit form. We can show that such canonical realisations have the internal model property. To this end, consider~(\ref{equ:IMProperty}) and let $\mathcal{G}_\xi(t) = F_{\mathrm{im}}(t) + G_{\mathrm{im}}(t) K_\xi(t)$ with $m = q$ and $K_\xi(t) = H_{\mathrm{im}}(t)$. If $\left(F_{\mathrm{im}}(\cdot), G_{\mathrm{im}}(\cdot), H_{\mathrm{im}}(\cdot)\right)$ is a canonical realisation of the internal model~(\ref{equ:IMsystem}), then~(\ref{equ:IMProperty}) is equivalent to~(\ref{equ:IMCanoniParamMoment}) with $\Pi_\xi = \Pi_M$, $\hat{\Phi}_\mathrm{im} = \Lambda$, and $\hat{\Xi}_\mathrm{im} = \Delta$. Now that we have established that the explicit canonical realization has the explicit internal model property, we show how to design a realization $\left(F_{\mathrm{im}}(\cdot), G_{\mathrm{im}}(\cdot), H_{\mathrm{im}}(\cdot)\right)$ which is canonical according to Definition~\ref{def:NonsmoothcanonicalReal}.
To this end, we first choose a time-invariant matrix pair $\left(F_{\mathrm{im}}, G_{\mathrm{im}}\right) \in \mathbb{R}^{m \times m} \times \mathbb{R}^{m \times 1}$ with $F_{\mathrm{im}}$ Hurwitz. Then the only remaining task is the design of $H_{\mathrm{im}}$ such that~(\ref{equ:IMCanoniParamPiAlgeb}) holds. By drawing inspiration from the design of a hybrid internal model unit~\cite{ref:marconi2013internal}, we propose a sufficient condition for designing a canonical realization in the following proposition.

\begin{prop}\label{prop:pseudoInverse}
Consider system~(\ref{equ:IMsystemNonsmooth}) with the explicit internal model property. Suppose $\hat{\Phi}_{\mathrm{im}}$ has the same properties as $\Lambda$ under Assumption~\ref{asmp:exoProp}. Let $m \geq s_2$ and select matrices $F_{\mathrm{im}}$ and $G_{\mathrm{im}}$ such that $F_{\mathrm{im}}$ is Hurwitz and the pair $\left(F_{\mathrm{im}}, G_{\mathrm{im}}\right)$ is controllable. If there exists a positive constant $\hat{t}$ and a positive integer $p_2 \leq s_2$ such that the solution $\Pi_{M}$ in~(\ref{equ:IMCanoniParamMoment}) satisfies $\operatorname{rank} (\Pi_{M}(t)) = p_2$ for all $t \geq \hat{t}$, then there exists a bounded piecewise continuous matrix-valued function $H_{\mathrm{im}}(t) \in \mathbb{R}^{1 \times m}$ such that $\hat{\Xi}_\mathrm{im}(t) = H_{\mathrm{im}}(t) \Pi_{M}(t)$ for all $t \geq \hat{t}$. In fact, it is possible to take $H_{\mathrm{im}}(t) = \hat{\Xi}_\mathrm{im}(t) \Pi_{M}(t)^{\dagger}$. 
\end{prop}

\begin{proof}
Let $\Psi_{M} := \Pi_M \hat{\Phi}_\mathrm{im}$. Since $\hat{\Phi}_{\mathrm{im}}$ is invertible, $\operatorname{rank} (\Psi_{M}(t)) = \operatorname{rank} (\Pi_{M}(t)) = p_2$ for all $t \geq \hat{t}$. Now denote $R(t) := \hat{\Xi}_\mathrm{im}(t)\hat{\Phi}_\mathrm{im}(t)$. Showing $H_{\mathrm{im}}(t) = \hat{\Xi}_\mathrm{im}(t) \Pi_{M}(t)^{\dagger}$ is equivalent to show $R(t) (I_{s_2} - \Psi_M(t)^{\dagger}\Psi_M(t)) = 0$ for all $t \geq \hat{t}$ with $\operatorname{rank}(\Psi_M) \equiv p_2 \leq s_2 \leq m$. Then this equation can be proved analogously to the proof of~\cite[Proposition 3]{ref:marconi2013internal}. Note that $H_{\mathrm{im}}$ is piecewise continuous as a consequence of the piecewise continuity of $\hat{\Xi}_\mathrm{im}$ and $\Pi_{M}$.
\end{proof}

\begin{remark}
    The idea of Proposition~\ref{prop:pseudoInverse} is similar to the design method of Kazantzis–Kravaris/Luenberger (KKL) observers for nonlinear systems~\cite{ref:kazantzis1998nonlinear,ref:kreisselmeier2003nonlinear,ref:andrieu2006existence}. More specifically, the KKL observer requires finding an immersion of a given nonlinear system into a dynamical system of the form
    \begin{equation*}
        \dot{z}_{o} = A_o z_o + By_o,  \qquad \hat{x}_o = T^{\mathrm{inv}}_o(z_o)
    \end{equation*}
    where $A_o$ is Hurwitz, $B_o$ is such that $(A_o, B_o)$ is controllable, and $T^{\mathrm{inv}}_o$ is a mapping designed to recover the nonlinear dynamics of the original system. In our work, similarly, the construction of a non-smooth canonical realisation of an explicit-form internal model~(\ref{equ:IMsystemNonsmooth}) boils down to finding a matrix-valued function $H_{\mathrm{im}}$ that recovers the non-smooth dynamics of the internal model~(\ref{equ:IMsystemNonsmooth}). 
\end{remark}

\subsection{Robust stabilisation}
\label{sec-robstab}
In this part, we show that when certain solvability conditions hold and the internal model unit admits a canonical realization, a standard high-gain feedback can be used to (robustly) stabilise the closed-loop system, \textit{i.e.}, condition \textbf{($\mathbf{S_E}$)} is satisfied. Considering the general case that the paper is focused on ($\omega$ is non-periodic, non-smooth, and can be diverging
), by Remark~\ref{remark:minphase}
the following assumption is introduced for sufficiently guaranteeing the existence of a solution to the regulator equations when the relative degree is one.
\begin{assumpt}\label{asmp:MinPhase}
    System~(\ref{equ:system}) is minimum phase.
\end{assumpt}

Now we study the stabilisation of the closed-loop system. To this end, we first assume that the system has a unitary relative degree, and then we discuss solutions to other cases. Note that without loss of generality, we assume that the high-frequency gain of the system $(A, B, C)$, denoted by $b$, is positive.
\begin{prop}\label{prop:stabilisation}
Consider system~(\ref{equ:system}). Suppose Assumptions~\ref{asmp:unitaryRD} and~\ref{asmp:MinPhase} hold. Given an internal model pair, either of the form~(\ref{equ:IMsystemSmooth}) or~(\ref{equ:IMsystemNonsmooth}), that admits a canonical realization $(F_{\mathrm{im}}, G_{\mathrm{im}}, H_{\mathrm{im}}(\cdot))$ with $F_{\mathrm{im}}$ Hurwitz and $H_{\mathrm{im}}(\cdot)$ bounded piecewise continuous, there exists a positive constant $\kappa$ such that for all $k > \kappa$, the regulator
\begin{equation}\label{equ:errorFdReguCanonical}
    \begin{aligned}
    \dot{\xi}(t) &= \left(F_{\mathrm{im}}+G_{\mathrm{im}} H_{\mathrm{im}}(t)\right) \xi(t)  -k G_{\mathrm{im}} e(t)\\
    u(t) &= H_{\mathrm{im}}(t) \xi(t) - k e(t)
    \end{aligned}
\end{equation}
with $\xi \in \mathbb{R}^{q}$ is such that condition \textbf{($\mathbf{S_E}$)} is satisfied.
\end{prop}
\begin{proof}
Define a linear transformation $T := [T_1^{\top},\;T_2^{\top}]^{\top}$ with $T_2 = C$ and $T_1$ being selected such that $T$ is non-singular. Then the change of variable $x \mapsto \theta :=T x = \operatorname{col}(z, \hat{y})$, with $z \in \mathbb{R}^{n-1}$ and $\hat{y} \in \mathbb{R}$, transforms system~(\ref{equ:system}) with $\omega \equiv 0$ into
\begin{equation}\label{equ:systemNormalForm}
    \begin{aligned}
    \dot{z}(t) &= A_{11}z(t) + A_{12}\hat{y}(t), \\ 
    \dot{\hat{y}}(t) &= A_{21}z(t) + A_{22}\hat{y}(t) + bu(t), \\
    e(t) &= \hat{y}(t).
    \end{aligned}
\end{equation}
where
\begin{equation*}
    T A T^{-1} =\left[\begin{array}{ll}
        A_{11} & A_{12} \\
        A_{21} & A_{22}
    \end{array}\right].
\end{equation*}
Then the change of variable $\xi \mapsto \chi = \xi - (1/b)G_{im}\hat{y}$ transforms systems~(\ref{equ:errorFdReguCanonical}) and~(\ref{equ:systemNormalForm}) into
\begin{equation*}
    \begin{aligned}
    & \dot{\chi}(t)=F_{\mathrm{im}} \chi(t) + \mathcal{A}_{\chi 1} z(t) + \mathcal{A}_{\chi 2} \hat{y}(t) \\
    & \dot{z}(t) = A_{11} z(t) +A_{12} \hat{y}(t) \\
    & \dot{\hat{y}}(t) = - b k\hat{y}(t) + A_y(t) \hat{y}(t) + b H_{\mathrm{im}}(t) \chi(t) + A_{21} z(t)
    \end{aligned}
\end{equation*}
with $A_{\chi 1} = -(1/b)G_{\mathrm{im}}A_{21}$, $\mathcal{A}_{\chi 2} = (1/b)(F_{\mathrm{im}}G_{\mathrm{im}} - G_{\mathrm{im}}A_{22})$, and $A_{y}(t) = A_{22} - bH_{\mathrm{im}}(t)G_{\mathrm{im}}$. By the boundedness of $H_{\mathrm{im}}$, the exponential stability of this closed-loop system can be proved by following the proof of the standard high-gain feedback design method~\cite[Lemma 1.5.4]{ref:isidori2003robust}.
\end{proof}

Proposition~\ref{prop:stabilisation} presents a stabilisation method when system~(\ref{equ:system}) has a unitary relative degree. Note that the stabilisation does not destroy the internal model property possessed by the canonical realisation, see~\cite[Proposition 3.3]{ref:zhang2009adaptive}. This implies that regulator~(\ref{equ:errorFdReguCanonical}) solves Problem~\ref{prob:ORErrorFd} if the internal model property (condition \textbf{($\mathbf{R_E}$)}) is satisfied by the internal model pair (and therefore the canonical realisation). In fact, this stabilisation method could be extended to a higher relative degree by means of high-gain observers or feedback transformations~\cite{ref:isidori2000tool}. For example, such extension has been introduced in the time-varying setting~\cite[Proposition 3.4]{ref:zhang2009adaptive}, which can be directly adapted to this non-smooth case.

Note that the stabilisation method discussed so far is robust, \textit{i.e.}, the methods still work even with the presence of system parameter uncertainties, so long as the sign of the high-frequency gain is known. Alternatively, a \textit{structually stable synthesis}~\cite{ref:zhang2009adaptive} can be constructed by regarding regulator~(\ref{equ:errorFdRegulator}) as the parallel interconnection of an internal model unit and a stabilisation unit. Therefore, the same stabilisation methods proposed in Proposition~\ref{prop:stabilisation} will be used in the next section, which is focused on the internal model principle.

Finally, in the case that system~(\ref{equ:system}) has a zero relative degree ($D \neq 0$), the same error-feedback output regulation problem can be solved by the use of an adaptive observer under an extra persistence of excitation condition, see~\cite{ref:niu2024adaptive}.

\section{Internal Model Principle}\label{sec:IMP}
The error-feedback controller proposed in Section~\ref{sec:ErrorFdProb} can solve the problem only when all the system parameters are perfectly known, \textit{i.e.}, there are no parameter uncertainties. This deficiency is due to the lack of robustness in the internal model. Thus, in this section, we formulate the robust non-smooth output regulation problem, and we solve the problem in the presence of model parameter uncertainty by looking into the construction of a robust internal model. To this end, we propose two design methods. First, we show that the augmentation-based internal model proposed in the hybrid regulation framework can also be extended to this non-smooth, non-periodic case. Due to the high-dimensional nature of the augmentation method, we then propose a new design method based on the concept of \textit{integral-based immersion}.

\subsection{Problem Formulation}
To study the problem with model parameter variations, we consider the uncertain LTI system of the form 
\begin{equation}\label{equ:UncertainSystem}
    \begin{aligned}
    \dot{x}(t) &= A(\mu)x(t) + B(\mu)u(t) + P(\mu)\omega(t),\\
    e(t) &= C(\mu)x(t) + Q(\mu)\omega(t),
    \end{aligned}
\end{equation}
with $\omega$ in~(\ref{equ:explicitGen}) and $\mu$ an unknown constant vector that belongs to a given compact set $\mathcal{K}_\mu \in \mathbb{R}^{n_\mu}$. Due to the presence of $\mu$, all model parameter values are not precisely known. In this context, we define the robust output regulation problem.

\begin{problem}[Robust Non-smooth Output Regulation Problem]\label{prob:ORRobust}
Consider system~(\ref{equ:UncertainSystem}) interconnected with the exosystem~(\ref{equ:explicitGen}) under Assumption~\ref{asmp:exoProp}. The output regulation problem consists in designing a regulator~(\ref{equ:errorFdRegulator}) such that for all $\mu \in \mathcal{K}_\mu$:
\begin{itemize}[leftmargin=28pt]
    \item[\textbf{($\mathbf{S_R}$)}] The origin of the closed-loop system obtained by interconnecting system~(\ref{equ:UncertainSystem}), the exosystem (\ref{equ:explicitGen}), and the regulator with $\omega(t) \equiv 0$ is exponentially stable.
    \item[\textbf{($\mathbf{R_R}$)}] The closed-loop system obtained by interconnecting system~(\ref{equ:UncertainSystem}), the exosystem~(\ref{equ:explicitGen}), and the regulator satisfies 
    $ \lim_{t \rightarrow +\infty} e(t) = 0, $
    uniformly for any $\left(w(t_0), x(t_0), \xi(t_0)\right) \in \mathbb{R}^{\nu+n+q}$.
\end{itemize}
\end{problem}

Since the model parameters are dependent on the uncertainty vector $\mu$, the robust output regulation problem normally cannot be solved by the error-feedback controller proposed in Section~\ref{sec:ErrorFdProb}. To see this, with the presence of $\mu$ in system~(\ref{equ:UncertainSystem}), we note that the solutions to the corresponding regulator equations of the form~(\ref{equ:reguEquation}), if they exist (for instance under the assumptions of Theorem~\ref{thm:solCondBoth}), are $\mu$-dependent. This means that if there exist bounded piecewise continuous mappings $\Pi^{*}_x: \mathbb{R} \times \mathbb{R}^{n_\mu} \to \mathbb{R}^{n \times \nu}$ and $\Delta^{*}: \mathbb{R} \times \mathbb{R}^{n_\mu} \to \mathbb{R}^{1 \times \nu}$ solving the uncertain regulator equations
\begin{equation}\label{equ:reguEquationUncertain}
    \begin{aligned}
    \Pi_x^{*}(t, \mu) &= \bigg(e^{A(\mu)(t-t_0)}\Pi_x^{*}(t_0)\\ &\!\!\!\!\!\!\!\!\!\!\!\!\!\!\!+\!\!
    \int_{t_0}^t \!\!\! e^{A(\mu)(t-\tau)}\!\! \left(P(\mu) \!+\! B(\mu)\Delta\!^{*}(\tau, \mu)\right)\! \Lambda(\tau, t_0) d \tau \!\!\bigg)\!  \Lambda(t, t_0)^{-1}\!\!, \\
    \!\!\!\!\bm{0}_{1 \times \nu} &= \lim_{t \rightarrow +\infty}C(\mu)\Pi_x(t, \mu) + Q(\mu),
    \end{aligned}
\end{equation}
for all $t \geq t_0$ and all $\mu \in \mathcal{K}_\mu$, then the internal model property, by Theorem~\ref{thm:IMProperty}, requires the design of $\mathcal{G}^{\mathrm{im}}_\xi(\cdot):=\mathcal{G}_\xi(\cdot)$ and $K^{\mathrm{im}}_\xi(\cdot):=K_\xi(\cdot)$ solving (\ref{equ:IMProperty}) for all $\mu \in \mathcal{K}_\mu$, namely solving
\begin{equation}\label{equ:IMPropertyRobust}
\begin{aligned}
    \frac{d}{dt}\left(\Pi^{*}_\xi(t, \mu)\Lambda(t, t_0) \right) &= \mathcal{G}^{\mathrm{im}}_\xi(t)\Pi^{*}_\xi(t, \mu)\Lambda(t, t_0), \\
    \Delta^{*}(t,\mu) &= K^{\mathrm{im}}_\xi(t)\Pi_\xi^{*}(t, \mu),
\end{aligned}
\end{equation}
for all $t \geq \hat{t}$ and all $\mu \in \mathcal{K}_\mu$, where $\Pi^{*}_{\xi}: \mathbb{R} \times \mathbb{R}^{n_\mu} \to \mathbb{R}^{q \times \nu}$ is bounded piecewise continuous. The satisfaction of the internal model property~(\ref{equ:IMPropertyRobust}) for all $\mu \in \mathcal{K}_\mu$ can be seen as the \textit{internal model principle} in this non-smooth output regulation setting. The principle implies that all controllers that solve the problem can reproduce an unknown error-zeroing input $u^{*}(t) = \Delta^{*}(t,\mu)\omega(t)$ as $t \rightarrow +\infty$, and, more specifically, should incorporate a robust internal model of the form
\begin{equation}\label{equ:IMsystemRobust}
\begin{aligned}  
    \omega(t) &= \Lambda(t, t_0) \omega_0, \\
    y_r(t) &= \Delta^{*}(t, \mu)\omega(t).
\end{aligned}
\end{equation}

This robust internal model property is normally not possessed by the error-feedback controller designed in the previous section. There are special cases, however, where the normal error-feedback controller is in fact robust. For example, if the solution to the uncertain regulator equations $\Delta^{*}$ satisfies
\begin{equation}\label{equ:ErrorFdRobustCond}
    \!\!\Delta\!^{*}\!(t, \mu)\Lambda(t, t_0) \!=\! \Delta(t)U(\mu)\Lambda(t, t_0) \!=\! \Delta(t)\Lambda(t, t_0)U(\mu)
\end{equation}
with $U(\mu) \in \mathbb{R}^{\nu \times \nu}$ and $\Delta$ the nominal solution to the regulator equations~(\ref{equ:reguEquation}) based on the nominal parameter values. 
Let $\mathcal{G}_\xi^{\mathrm{im}}(\cdot)$, $K_{\xi}^{\mathrm{im}}(\cdot)$ be
such that~(\ref{equ:IMProperty})
admits a bounded piecewise continuous solution $\Pi_\xi$, then it is straightforward to see that $\Pi^{*}_\xi(t, \mu) = \Pi_\xi(t) U(\mu)$ is a solution to~(\ref{equ:IMPropertyRobust}). However, condition~(\ref{equ:ErrorFdRobustCond}) is quite restrictive due to the commutativity requirement between $\Lambda$ and $U$. The rest of the section will then propose methods of establishing a valid robust internal model unit for more general cases.


\subsection{Augmentation-based Internal Model}
\label{sec:HybridIMP}
As mentioned in Section~\ref{sec:intro}, existing studies into the output regulation problems that involve non-smooth dynamics of systems are conducted in the hybrid framework with extra assumptions on the periodicity of the exogenous signal, see~\cite{ref:marconi2013internal}. Now we show that the method of constructing a robust internal model proposed in that framework can be extended to our non-periodic setting. 

The robust internal model generalised from the hybrid framework is based on how the uncertain parameters $\mu$ of system~(\ref{equ:UncertainSystem}) influence the solution to the regulator equations. More specifically, it assumes that the solution to the uncertain regulator equations~(\ref{equ:reguEquationUncertain}), $\Delta^{*}(t, \mu)$, satisfies
\begin{equation}\label{equ:HybridUncertAssump}
    \Delta^{*}(t, \mu) = \Delta^{\prime}(t)U^{\prime}(\mu),
\end{equation}
with $\Delta^{\prime}(t) \in \mathbb{R}^{1 \times l}$, $l > 0$ a known vector function and $U^{\prime}(\mu) \in \mathbb{R}^{l \times \nu}$ an uncertain matrix. Compared with~(\ref{equ:ErrorFdRobustCond}), this hypothesis~(\ref{equ:HybridUncertAssump}) does not rely on any commutativity properties and is less restrictive. If~(\ref{equ:HybridUncertAssump}) holds, an augmented system of the form
\begin{equation}\label{equ:IMsystemAugRobust}
\begin{aligned}  
    \hat{\omega}(t) &= \left(I_{\nu l} \otimes \Lambda(t, t_0)\right) \hat{\omega}_0, \\
    y_\omega(t) &= \Delta^{\prime}(t)\left(I_{l} \otimes (\operatorname{vec}I_\nu)^{\top}\right)\hat{\omega}(t),
\end{aligned}
\end{equation}
is able to generate the output of system~(\ref{equ:IMsystemRobust}) when initialised by $\hat{\omega}_0 = \left(\operatorname{vec} (U^{{\prime}\top}) \otimes I_\nu\right) w_0$. The reason for this choice is justified by the following proposition.
\begin{prop}
Let the pair $(\mathcal{G}^{\mathrm{im}}_\xi(\cdot), K^{\mathrm{im}}_\xi(\cdot))$  be such that
\begin{equation*}
    \begin{aligned}
    \frac{d}{dt}\left(\Pi^{\prime}_\xi(t)(I_{\nu l} \otimes \Lambda(t, t_0)) \right) &= \mathcal{G}^{\mathrm{im}}_\xi(t)\Pi^{\prime}_\xi(t)(I_{\nu l} \otimes \Lambda(t, t_0)) \\
    \Delta^{\prime}(t)\left(I_{l} \otimes\left(\operatorname{vect} I_\nu\right)^\top\right) &= K^{\mathrm{im}}_\xi(t) \Pi^{\prime}_\xi(t)   
    \end{aligned}
\end{equation*}
for some $\Pi^{\prime}_\xi:\mathbb{R} \to \mathbb{R}^{q \times \nu^2 l}$. Then,
\begin{equation*}
    \Pi_\xi^{*}(t, \mu)=\Pi^{\prime}_\xi(t)\left(\operatorname{vec} (U^{{\prime}\top}) \otimes I_\nu\right)
\end{equation*}
is a solution to (\ref{equ:IMPropertyRobust}).
\end{prop}
\begin{proof}
    See the proof of~\cite[Proposition 6]{ref:marconi2013internal}.
\end{proof}

Therefore, if~(\ref{equ:HybridUncertAssump}) holds with $\Delta^{\prime}$ known, the design of a regulator~(\ref{equ:errorFdReguCanonical}) solving Problem~\ref{prob:ORRobust} follows the same way discussed in Section~\ref{sec:ErrorFdProb} with $\hat{\Phi}_{\mathrm{im}}(t) = I_{\nu l} \otimes \Lambda(t, t_0)$ and $\hat{\Xi}_{\mathrm{im}}(t) = \Delta^{\prime}(t)(I_{l} \otimes (\operatorname{vec}I_\nu)^{\top})$.

This augmentation-based internal model has two drawbacks: the first is that many uncertainties do not satisfy the structure in (\ref{equ:HybridUncertAssump}); the second is that the method produces controllers of high dimension, \textit{i.e.}, of
at least $\nu^2 l$. The hybrid study~\cite{ref:marconi2013internal} suggested that this dimension can be reduced depending on how the uncertainties $\mu$ affect the entries of $U^{\prime}(\mu)$. We adapt here that discussion as the resulting reduced controller will be implemented for comparison in the simulations in Section~\ref{sec:example}. Define a $\mu$-independent matrix $L \in \mathbb{R}^{\nu l \times l^{\prime}}$ with $l^\prime > 0$ such that
\begin{equation}\label{equ:AugMethodRedOrder}
   \bigcup_{\mu \in \mathcal{K}_\mu} \operatorname{span} \operatorname{vec}(U^{\prime}(\mu)) \subseteq \operatorname{Im} L. 
\end{equation}
Then for all $\mu \in \mathcal{K}_\mu$, there exist a vector $c_\mu(\mu) \in \mathbb{R}^{l^{\prime}}$ and a matrix $\bar{L} \in \mathbb{R}^{l \times \nu l^{\prime}}$ such that $\operatorname{vec}(U^{\prime}(\mu)) = Lc_\mu(\mu)$ and $U^{\prime}(\mu) = \bar{L}(c_\mu(\mu) \otimes I_{\nu})$. Then any output of system~(\ref{equ:IMsystemRobust}) can be generated by the augmented system of the form~(\ref{equ:IMsystemAugRobust}) with $I_{\nu l} \otimes \Lambda(t, t_0)$ and $I_{l} \otimes (\operatorname{vec}I_\nu)^{\top}$ replaced by $I_{l^{\prime}} \otimes \Lambda(t, t_0)$ and $\bar{L}$ when initialised by $\hat{\omega}_0 = (c_\mu(\mu) \otimes I_{\nu})\omega_0$. As such, if not all elements of $U^{\prime}(\mu)$ vary independently with $\mu$, then the controller dimension can be reduced to $\nu l^{\prime}$ with $l^{\prime} < \nu l$. 

However, depending on the structure of $U'(\mu)$, the regulator dimension established by the augmentation method can still be large in many cases. To eliminate this dimensional issue, in the next section we introduce another method of designing a robust internal model, which can have a smaller regulator dimension than the augmentation method.

\subsection{Immersion-Based Robust Internal Model}
\label{sec:ImmersionIMP}
In this section we develop an immersion-based robust internal model for the non-smooth case. Before doing so let us recall that
in the smooth linear case the expected internal model unit can naturally be immersed into an uncertainty-independent observable system by exploiting the Cayley-Hamilton theorem. Similarly, in the smooth nonlinear case the expected nonlinear internal model unit is assumed to admit an immersion into a linear observable system~\cite[Section 8.4]{ref:isidori1995nonlinear}. The concept of immersion was also extended to the linear time-varying case~\cite{ref:zhang2009adaptive}, which is classified into different types (strong, regular, weak) depending on the observability properties of the system obtained after immersion. Note that in both nonlinear and linear time-varying smooth cases, the immersion assumption is generally based on checking successive derivatives of the expected output of the internal model. This kind of differential-based hypothesis, however, is not pursuable in our case due to the non-smoothness of $\Delta^{*}$ and $\Lambda$ in the internal model~(\ref{equ:IMsystemRobust}). With this restriction, we first introduce the concept of immersion in the non-smooth case, and we then propose a novel \textit{integral-based} sufficient condition for such immersion to exist.

\begin{define}[Explicit Immersion]
The family of explicit-form autonomous systems~(\ref{equ:IMsystemRobust}) is said to be \textit{immersed} into an implicit-form system~(\ref{equ:IMsystemSmooth}) with piecewise continuous functions $\tilde{\Phi}_{\mathrm{im}}: \mathbb{R} \to \mathbb{R}^{s_1 \times s_1}$ and $\tilde{\Xi}_{\mathrm{im}}: \mathbb{R} \to \mathbb{R}^{1 \times s_1}$ if there exist a positive constant $\hat{t}$ and an absolutely continuous mapping $\Psi_\Upsilon: \mathbb{R} \times \mathbb{R}^{n_\mu} \to \mathbb{R}^{s_1 \times s_1}$ such that
\begin{equation}\label{equ:immersionMapping}
\begin{aligned}
    \dot{\Psi}_\Upsilon(t, \mu) &= \tilde{\Phi}_{\mathrm{im}}(t)\Psi_\Upsilon(t, \mu), \\
    \Delta^{*}(t, \mu)\Lambda(t, t_0) &= \tilde{\Xi}_{\mathrm{im}}(t)\Psi_\Upsilon(t, \mu),
\end{aligned}
\end{equation}
for all $t \geq \hat{t}$ and all $\mu \in \mathcal{K}_\mu$.
\end{define}

This explicit immersion suggests that the output response of an internal model of the form~(\ref{equ:IMsystemSmooth}) can produce any output response $y_r$ generated by system~(\ref{equ:IMsystemRobust}). In comparing (\ref{equ:immersionMapping}) with (\ref{equ:IMPropertyRobust}), we note that the term $\Psi_\Upsilon(t, \mu)$ plays the role of $\Pi^{*}_\xi(t, \mu)\Lambda(t, t_0)$, and consequently the term $\Upsilon(t, \mu) := \Psi_\Upsilon(t, \mu)\Lambda(t, t_0)^{-1}$ replaces $\Pi^{*}_\xi(t, \mu)$. Differently from (\ref{equ:IMPropertyRobust}), the expression (\ref{equ:immersionMapping}) does not require the boundedness of $\Upsilon(t, \mu)$.
We will show later that, by finding the canonical realization of the pair $(\tilde{\Phi}_{\mathrm{im}}(\cdot), \tilde{\Xi}_{\mathrm{im}}(\cdot))$, the unboundedness of $\Upsilon$ does not affect the existence of a bounded solution $\Pi^{*}_\xi$ to~(\ref{equ:IMPropertyRobust}), so long as $\tilde{\Phi}_{\mathrm{im}}$ and $\tilde{\Xi}_{\mathrm{im}}$ are bounded.

We recall, also, that in the time-varying smooth case, the pair $(\tilde{\Phi}_{\mathrm{im}}(\cdot), \tilde{\Xi}_{\mathrm{im}}(\cdot))$ can be uniformly observable (which corresponds to \textit{strong immersion} defined in that case). This uniform observability was important as it was shown to be equivalent to the existence of a transformation of the internal model~(\ref{equ:immersionMapping}) into an observability canonical form~\cite{ref:shieh1987transformations}, which, in that smooth context, was crucial for finding a valid canonical realization. However, the uniform observability is no longer satisfied by $\tilde{\Phi}_{\mathrm{im}}$ and $\tilde{\Xi}_{\mathrm{im}}$ in the non-smooth case, as $\tilde{\Phi}_{\mathrm{im}}$ and $\tilde{\Xi}_{\mathrm{im}}$ are not differentiable and do not admit an observability canonical form due to the piecewise continuity of $\Lambda$ and $\Delta^{*}$. As a consequence, we omit the discussion of the observability of the pair $(\tilde{\Phi}_{\mathrm{im}}(\cdot), \tilde{\Xi}_{\mathrm{im}}(\cdot))$ in the definition of immersion. 

Now we provide a sufficient condition for the existence of such explicit immersion. For brevity, denote $R^{*}(t, \mu) = \Delta^{*}(t,\mu)\Lambda(t, t_0)$.
\begin{lem}\label{lem:immersion}
Suppose Assumption~\ref{asmp:exoProp} holds. The non-smooth autonomous system~(\ref{equ:IMsystemRobust}) admits an explicit immersion if there exist a positive constant $\hat{t}$, a positive integer $d$, and bounded piecewise continuous functions $a_1(t), a_2(t), \ldots, a_{d}(t)$ such that\footnote{Recall that $\mathcal{I}^{[i]}_{t_0}$ indicates the iterated integral defined in the notation.}
\begin{equation}\label{equ:immersion}
    R^{*}(t, \mu) + \sum_{i=1}^{d} a_i(t) \mathcal{I}^{[i]}_{t_0}[R^{*}(t, \mu)] = \bm{0}_{1 \times \nu},
\end{equation}
for all $t \geq \hat{t}$ and all $\mu \in \mathcal{K}_\mu$.
\end{lem}
\begin{proof}
The satisfaction of~(\ref{equ:immersion}) implies that if the pair $(\tilde{\Phi}_{\mathrm{im}}(\cdot),\; \tilde{\Xi}_{\mathrm{im}}(\cdot))$ are in the form
\begin{equation}\label{equ:immersionPair}
    \begin{aligned}
    \tilde{\Phi}_{\mathrm{im}}(t) & =\left[\begin{array}{cccc}
    0 & 1 & \cdots & 0 \\
    \vdots & \vdots & \ddots & \vdots \\
    0 & 0 & \cdots & 1 \\
    -a_d(t) & -a_{d-1}(t) & \cdots & -a_1(t)
    \end{array}\right], \\
    \tilde{\Xi}_{\mathrm{im}}(t) & =\left[\begin{array}{llll}
    -a_d(t) & -a_{d-1}(t) & \cdots & -a_1(t)
    \end{array}\right],
    \end{aligned}
\end{equation}
then there exists a mapping
\begin{equation}\label{equ:ImmersionSolution}
    \Psi_{\Upsilon}(t, \mu)\!=\!\!
    \left[\begin{array}{c}
    \mathcal{I}^{[d]}_{t_0}[R^{*}(t,\mu)] \\
    \mathcal{I}^{[d-1]}_{t_0}[R^{*}(t,\mu)] \\
    \vdots \\
    \mathcal{I}^{[1]}_{t_0}[R^{*}(t,\mu)]
    \end{array}\right],  
\end{equation}
that solves~(\ref{equ:immersionMapping}) for all $t \geq \hat{t}$ and all $\mu \in \mathcal{K}_\mu$.
\end{proof}

Lemma~\ref{lem:immersion} provides a way of finding the immersion based on integral relations when the internal model output dynamics $y_r$ in~(\ref{equ:IMsystemRobust}) is non-smooth. The analytical solution to (\ref{equ:immersion}), namely ($a_1, a_2, \ldots, a_{d}$), may be difficult (or impossible) to characterize, similarly to the strong immersion in the smooth case. However, it can be derived in some specific cases. For example, if~(\ref{equ:HybridUncertAssump}) is satisfied (\textit{i.e.}, the requirement of the augmentation-based internal model design), then 
an immersion-based internal model may be constructed. This is illustrated with an example. 
\begin{example}\label{eg:Immersion}
Consider a non-smooth autonomous system~(\ref{equ:IMsystemRobust}) with $t_0 = 0$, $\Lambda(t, 0) = \operatorname{diag}\left(\lambda_1(t, 0),\; \lambda_2(t, 0)\right)$ and $\Delta^{*}(t, \mu) = \left[\begin{array}{ll} \mu_1 \sin(t^{3/2}) - \mu_3 & \mu_2 \sqcap(t) \end{array}\right]$, where $\lambda_1(t, 0) = \sin(t^{3/2})+2$, $\lambda_2(t, 0) = \nabla(t) - 3$, 
$\sqcap(t) := \operatorname{sign}(\sin (t))$ is a periodic square wave, $\nabla(t) := \frac{4}{T}\int_0^t \sqcap(\tau)d\tau-1$ is a periodic triangular wave, and $(\mu_1, \mu_2) \in \mathbb{R} \times \mathbb{R}$. Note that $\Delta^{*}$ in this case satisfies~(\ref{equ:HybridUncertAssump}) with $\Delta^{\prime}(t) = [\sin(t^{3/2}),\; \sqcap(t),\; 1]$ and
$$
U^{\prime}(\mu) = \left[\begin{array}{cc}
    \mu_1 & 0 \\
    0 & \mu_2 \\
    -\mu_3 & 0
    \end{array}\right],
$$
which implies that the problem is not solvable by the non-robust error-feedback controller proposed in Section~\ref{sec:ErrorFdProb}, but can be solved by an augmentation-based internal model with a reduced dimension $\nu l^{\prime} = 6$ because $\nu = 2$ and $l^{\prime} = 3$. The problem can also be solved by the immersion-based method that, as we will see, will have a smaller internal model dimension. Let $F(t) = [\sin(t^{3/2})\lambda_1(t),\; \lambda_1(t),\; \sqcap(t)\lambda_2(t)]$ so that
\begin{equation*}
    R^{*}(t, \mu) = \Delta^{*}(t, \mu)\Lambda(t, 0) = F(t) \left[\begin{array}{rc}
    \mu_1 & 0 \\
    -\mu_3 & 0 \\
    0 & \mu_2
    \end{array}\right].
\end{equation*}
A valid explicit immersion can be established if~(\ref{equ:immersion}), with $R^{*}(t, \mu)$ replaced by $F(t)$, is satisfied by a group of bounded piecewise continuous functions $a_1, a_2, \ldots, a_{d}$. These functions, grouped as $v(t) = \operatorname{col}(a_{d}(t), a_{d-1}(t), \cdots, a_{1}(t))$, can be found if there exist a positive integer $d$ and a positive constant $\hat{t}$ such that for all $t \geq \hat{t}$,
\begin{equation}\label{equ:rankCompare}
    \operatorname{rank}\left(\hat{\Psi}_{\Upsilon}(t)^{\top}\right) = \operatorname{rank}\left(\left[\begin{array}{cc}
    \hat{\Psi}_{\Upsilon}(t)^{\top}  &  F(t)^{\top} 
    \end{array}\right]\right),
\end{equation}
where $$\hat{\Psi}_{\Upsilon}(t) \!=\! \left[\mathcal{I}^{[d]}_{0}[ F(t)]^{\top}\!,\;  \mathcal{I}^{[d-1]}_{0}[ F(t)]^{\top}\!,\; \cdots,\; \mathcal{I}^{[1]}_{0}[F(t)]^{\top}\right]^{\top}\!\!\!.$$ Then the functions $a_{i}$'s are the components of $v(t) = (\hat{\Psi}_{\Upsilon}(t)^\top)^{\dagger}F(t)^{\top}$. In this example, these functions exist when $d = 4$ and $\hat{t} = 5$. Their trajectories are depicted in Fig.~\ref{fig:Immersion_Example}. Since $d < \nu l^{\prime}$ in this case, the example shows that the immersion-based internal model can have a smaller dimension than the augmentation-based internal model. \ \hfill $\square$
\begin{figure}[t]
	\begin{centering}
		\includegraphics[width=\linewidth]{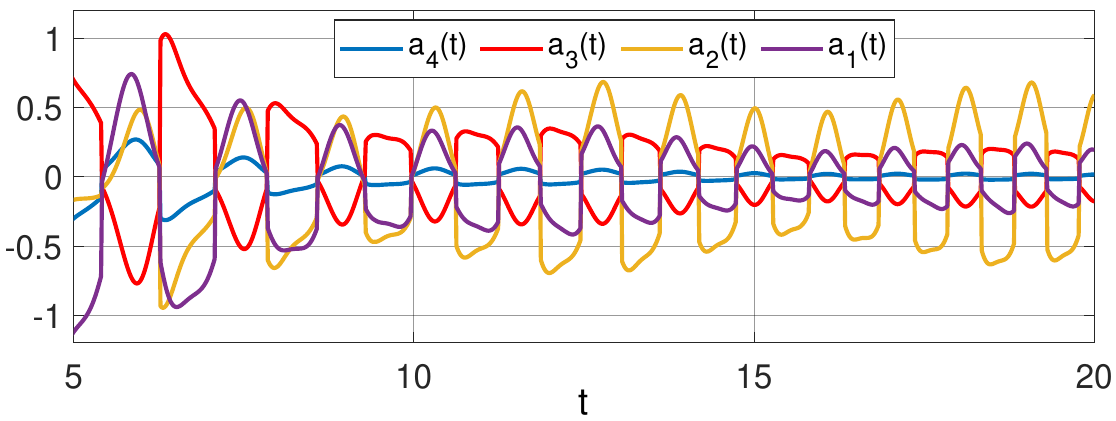}
		\caption{Time histories of the functions $a_1(t)$, $a_2(t)$, $a_3(t)$, and $a_4(t)$.}
		\label{fig:Immersion_Example}
	\end{centering}
\end{figure}
\end{example}


Example~\ref{eg:Immersion} shows that the immersion-based method can have a dimension lower than the argumentation-based regulator in some cases. However, this does not mean that the immersion-based method is always dimension-efficient, as the value of $d$ and the satisfaction of~(\ref{equ:rankCompare}) depend on the structures of both $\Delta^{*}$ and $\Lambda$.



Now we discuss how to find a canonical realisation of an explicit immersion established by the pair $(\tilde{\Phi}_{\mathrm{im}}(\cdot), \tilde{\Xi}_{\mathrm{im}}(\cdot))$ in~(\ref{equ:immersionPair}). As aforementioned, the pair $(\tilde{\Phi}_{\mathrm{im}}(\cdot), \tilde{\Xi}_{\mathrm{im}}(\cdot))$ does not admit an observability canonical form. Thus, the 
method to find a canonical realisation~(\ref{equ:canonicalRealSylv}) of an internal model of the form~(\ref{equ:IMsystemSmooth}) based on the observability canonical form in the smooth setting is no longer valid. However, inspired by Proposition~\ref{prop:pseudoInverse}, a realisation can be found by the following corollary.
\begin{corol}\label{corol:pseudoInverse}
Suppose Assumption~\ref{asmp:exoProp} holds. Let $m \geq s_1$ and select matrices $F_{\mathrm{im}}, G_{\mathrm{im}}$ such that $F_{\mathrm{im}}$ is Hurwitz and the pair $\left(F_{\mathrm{im}}, G_{\mathrm{im}}\right)$ is controllable. If there exist a positive constant $\hat{t}$ and a positive integer $p_1 \leq s_1$ such that the solution $M$ to
\begin{equation}\label{equ:canonicalRealSylvComb}
    \dot{M}(t)+M(t) \tilde{\Phi}_{\mathrm{im}}(t) =F_{\mathrm{im}}M(t) + G_{\mathrm{im}} \tilde{\Xi}_{\mathrm{im}}(t)
\end{equation}
satisfies $\operatorname{rank} (M(t)) = p_1$ for all $t \geq \hat{t}$, then there exists a bounded piecewise continuous matrix-valued function $H_{\mathrm{im}}(t) \in \mathbb{R}^{1 \times m}$ such that $\tilde{\Xi}_{\mathrm{im}}(t) = H_{\mathrm{im}}(t) M(t)$ for all $t \geq \hat{t}$. In fact, it is possible to take $H_{\mathrm{im}}(t) = \tilde{\Xi}_{\mathrm{im}}(t) M(t)^{\dagger}$. 
\end{corol}
\begin{proof}
    The proof is analogous to the proof of Proposition~\ref{prop:pseudoInverse} and is omitted.
\end{proof}


The last matter to clarify here is that, although the satisfaction of the explicit immersion~(\ref{equ:immersionMapping}) does not imply the existence of a bounded $\Pi^{*}_\xi$ solving~(\ref{equ:IMPropertyRobust}), the existence of a valid canonical realisation $\left(F_{\mathrm{im}}(\cdot), G_{\mathrm{im}}(\cdot), H_{\mathrm{im}}(\cdot)\right)$ of the pair $(\tilde{\Phi}_{\mathrm{im}}(\cdot), \tilde{\Xi}_{\mathrm{im}}(\cdot))$ in~(\ref{equ:immersionMapping}) does in fact imply the existence of a bounded $\Pi^{*}_\xi$ solving~(\ref{equ:IMPropertyRobust}). To show this, let $\mathcal{G}^{\mathrm{im}}_\xi(t) = F_\mathrm{im}(t) + G_\mathrm{im}(t)H_\mathrm{im}(t)$, $K^{\mathrm{im}}_\xi(t) = H_\mathrm{im}(t)$, and $\Psi^{*}_\xi(t, \mu) = M(t)\Psi_\Upsilon(t, \mu)$ with $M$ solving~(\ref{equ:canonicalRealSylvComb}) and $\Psi_\Upsilon(t, \mu)$ solving~(\ref{equ:immersionMapping}). Then, 
\begin{equation*}
    \begin{aligned}
        \dot{\Psi}^{*}_\xi(t, \mu) &= \dot{M}(t)\Psi_\Upsilon(t, \mu) + M(t)\dot{\Psi}_\Upsilon(t, \mu) \\
        &= (F_\mathrm{im}(t)M(t) + G_\mathrm{im}(t)\tilde{\Xi}_\mathrm{im}(t))\Psi_\Upsilon(t, \mu) \\
        &= F_\mathrm{im}(t)\Psi^{*}_\xi(t, \mu) + G_\mathrm{im}(t)\Delta^{*}(t, \mu)\Lambda(t, t_0)
    \end{aligned}
\end{equation*}
with $H_\mathrm{im}(t)\Psi^{*}_\xi(t, \mu) = \tilde{\Xi}_\mathrm{im}(t)\Psi_\Upsilon(t, \mu) = \Delta^{*}(t, \mu)\Lambda(t, t_0)$. Therefore, it is straightforward to see that $\Pi^{*}_\xi(t, \mu) = \Psi^{*}_\xi(t, \mu)\Lambda(t, t_0)^{-1} = M(t)\Psi_\Upsilon(t, \mu)\Lambda(t, t_0)^{-1}$ is the solution to~(\ref{equ:IMPropertyRobust}) and can be expressed by
\begin{equation}\label{equ:IMPropertyRobustSolution}
\begin{aligned}
    \Pi^{*}_\xi(t, \mu) &= \bigg(\Phi_F(t,\hat{t})\Pi^{*}_\xi(\hat{t}, \mu)\\
    &\!\!\!\!\!\!\!\!\!\!\!+\!\! 
    \int_{\hat{t}}^t \!\! \Phi_F(t,\tau)G_{\mathrm{im}}(\tau)\Delta(\tau, \mu) \Lambda(\tau, t_0) d \tau \!\bigg)\! \Lambda(t, t_0)^{-1}\!\!\!,
\end{aligned}
\end{equation}
for all $t \geq \hat{t}$, where $\Phi_{F}(t, t_0)$ is the transition matrix of system~(\ref{equ:FimSystem}). Since system~(\ref{equ:FimSystem}) is exponentially stable, Lemma~\ref{lem:moment} implies that under Assumption~\ref{asmp:exoProp}, $\Pi^{*}_\xi$ in~(\ref{equ:IMPropertyRobustSolution}) is bounded and piecewise continuous for all times.

\section{Illustrative Example}\label{sec:example}
In this section, we provide an example to illustrate the aforementioned theory. We consider a circuit regulation problem, 
and we first show that, even when the exogenous signal is diverging, the solutions to the regulator equations derived in Section~\ref{sec:FullInfoProb} are bounded. Then, we look at the problem with model parameter uncertainties, and we design two robust regulators using the augmentation-based method and the immersion-based method proposed in Section~\ref{sec:IMP}.

Consider the circuit depicted in Fig.~\ref{fig:Example_Circuit} consisting of an input voltage source $u(t)$ connected, via an input resistor $R_{in}$, to an RLC circuit composed of the series connection of an inductor $L_r$, a capacitor $C_L$, and a resistor $R_L$. The capacitor $C_L$ and resistor $R_L$ belong to an external load enclosed by a red dashed box in Fig~\ref{fig:Example_Circuit}. Meanwhile, the RLC circuit is also subject to the influence of an external circuit, which is equivalently represented by a voltage generator $d_{th}(t)$ in series with a resistor $R_{th}$ using Thévenin's theorem. The disturbance signal is assumed to be a time-varying rectangular wave $d_{th}(t) = (\sqcap(\frac{2 \pi}{3}t^{\frac{3}{2}}+\frac{\pi}{2}) + 2)\hat{d}_0$ where $\sqcap(t) := \operatorname{sign}(\sin (t))$. The aim of this example is to design the input voltage $u(t)$ to regulate the voltage of the load, $y(t)$, to track a given reference signal $r_{f}(t)$, which is assumed to be a diverging triangular wave $r_{f}(t) = (t+1)\nabla(\frac{2 \pi}{3} t +\frac{\pi}{2})\hat{r}_0$ where $\nabla(t) := \frac{4}{T}\int_0^t \sqcap(\tau)d\tau-1$. Note that the initial conditions $\hat{d}_0$ and $\hat{r}_0$ are unknown, while we pick $t_0 = 0$. 

\begin{figure}[t]
	\begin{centering}
		\includegraphics[width=0.8\linewidth]{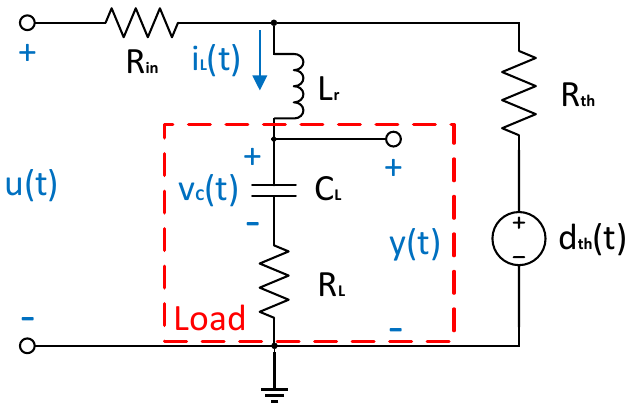}
		\caption{Schematics of the RLC circuit.}
		\label{fig:Example_Circuit}
	\end{centering}
\end{figure}
\begin{figure}[t]
	\begin{centering}
		\includegraphics[width=\linewidth]{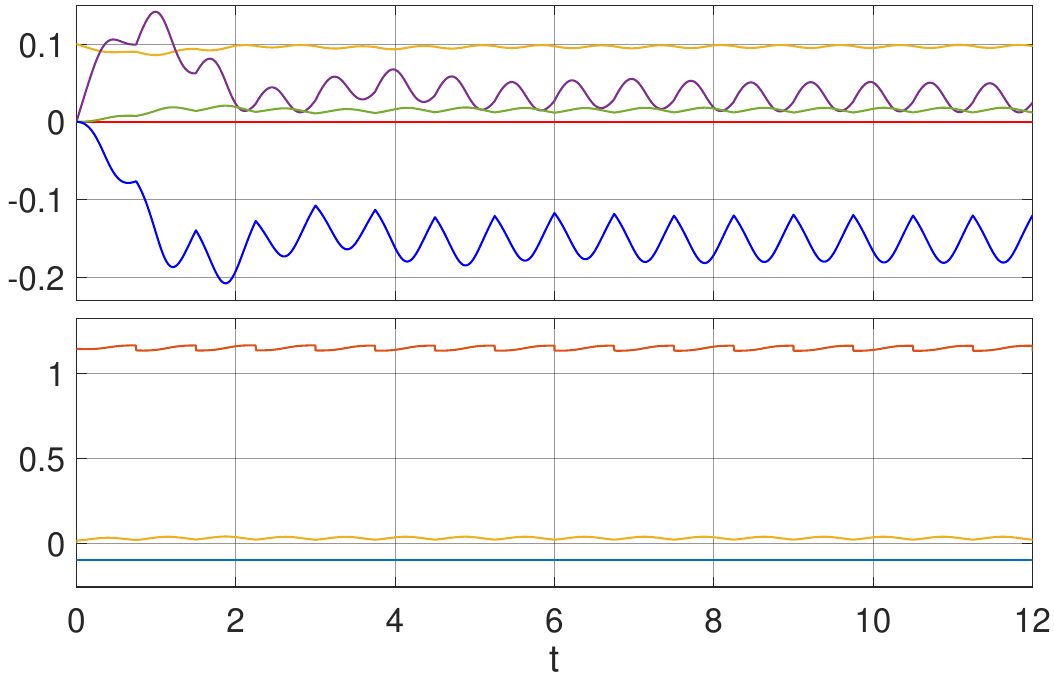}
		\caption{\textbf{Top:} time history of solution $\Pi_x(t)$ of the regulator equations. \textbf{Bottom:} time history of solution $\Delta(t)$ of the regulator equations.}
		\label{fig:Regulator_Equations}
	\end{centering}
\end{figure}

We first design the explicit-form exosystem that generates the exogenous signals $d_{th}$ and $r_{f}(t)$. Let $\Lambda_d(t, 0) = \sqcap(\frac{2 \pi}{3}t^{\frac{3}{2}}+\frac{\pi}{2}) + 2$ and $\Lambda_r(t, 0) = (t+1)\nabla(\frac{2 \pi}{3} t +\frac{\pi}{2})$. Following the proof of Proposition~\ref{prop:LambdaDesign}, we can define the exogenous signal $\omega(t) = [d_{th}(t),\; r_{f}(t),\;0]^{\top} = \Lambda(t, 0)\omega_0$ where $\Lambda = \text{blkdiag}(\Lambda_d,\; \hat{\Lambda}_r)$ with
\begin{equation}\label{equ:lambdaBlock}
        \hat{\Lambda}_r(t, 0) = \left[\begin{array}{rr}
        \Lambda_r(t, 0) & -\tilde{\Lambda}_r(t, 0) \\
        \tilde{\Lambda}_r(t, 0) & \Lambda_r(t, 0)
        \end{array}\right],
\end{equation}
where $\tilde{\Lambda}_r(t, 0) = (t+1)\nabla(\frac{2 \pi}{3} t)$ is selected to guarantee the invertibility of $\hat{\Lambda}_r$ (and therefore $\Lambda$). Then we are ready to derive the input-output dynamics of the circuit by a state-space model of the form~(\ref{equ:system}). By denoting $x_1(t)$ the current across the inductor $i_{L}(t)$, $x_2(t)$ the voltage across the capacitor $v_{C}(t)$, and $e(t) = y(t) - r_{f}(t)$ the regulation error, the circuit can be described by system~(\ref{equ:system}) with
\begin{equation*}
    \begin{aligned}
        A &=\left[\begin{array}{cc}
         \varepsilon_2   &   -\frac{1}{L_r} \\
         \frac{1}{C_L}   &  0 
        \end{array}\right]\!,  \quad
        B=\left[\begin{array}{c}
           R_{th}\varepsilon_1 \\
           0
        \end{array}\right]\!,\quad
        C = \left[\begin{array}{cc}
           R_{L} \\
           1
        \end{array}\right]^{\top}\!\!\!, \\
        D &= 0, \qquad
        P = \!\left[\begin{array}{cc}
             P_d & 0_{2 \times 2}\\
        \end{array}\right]\!, \qquad
        Q =\! \left[\begin{array}{ccc}
           0 & -1 & 0\\
        \end{array}\right]\!,
    \end{aligned}
\end{equation*}
where $\varepsilon_1 = ((R_{in}+ R_{th})L_r)^{-1}$, $\varepsilon_2 = -R_L/L_r - R_{in}R_{th}\varepsilon_1$, and $P_d = [\begin{array}{cc} R_{in}\varepsilon_1 & 0\end{array}]^{\top}$. In the following simulations, we set the nominal values of the electrical components as follows: $R_{in} = 0.5\;\Omega$, $R_{th} = 5\;\Omega$, $L_r = 0.1\;H$, $C_L = 0.3\;F$, and $R_L = 10\;\Omega$. When the real values of all electrical components are identical to their nominal values, $Q\Lambda$ satisfies Assumption~\ref{asmp:pwDiff}, and the system is minimum phase with a unitary relative degree. By Theorem~\ref{thm:solCondBoth}, the problem is solvable and the solutions to the regulator equations can be numerically computed by MATLAB Simulink. Fig.~\ref{fig:Regulator_Equations} depicts the time histories of solutions $\Pi_x$ (top graph) and $\Delta$ (bottom graph) of the regulator equations. The figure shows that the solutions to the regulator equations are bounded and piecewise continuous, even when $\Lambda$ is diverging.

\begin{figure}[t]
	\begin{centering}
		\includegraphics[width=\linewidth,height=12cm]{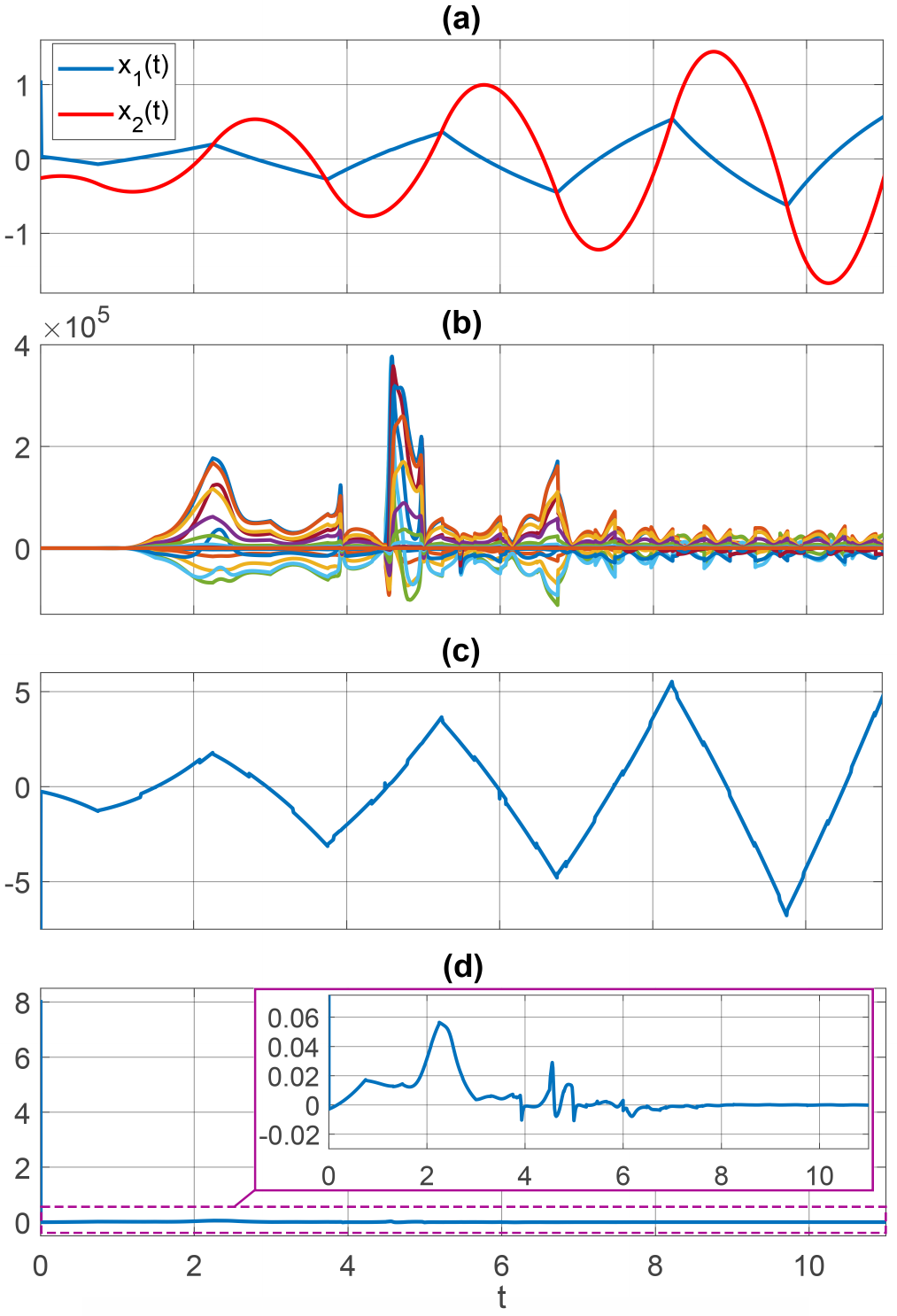}
		\caption{\textbf {(a)} Time history of system state $x(t)$. \textbf{(b)} Time history of $H_{\mathrm{im}}(t)$ for the canonical realization. \textbf{(c)} Time history of the input $u(t)$ generated by the robust augmentation-based regulator. \textbf{(d)} Time history of the regulation error $e(t)$.}
		\label{fig:Augmentation_Result}
	\end{centering}
\end{figure}

Now we consider the problem in which the true values of the load capacitor $C_{L}$ and load resistor $R_{L}$ are unknown. The only available information is that they vary around their nominal values: $C_L = \hat{C}_L \pm \mu_C$ and $R_L = \hat{R}_L \pm \mu_R$ with nominal values $\hat{C}_L = 0.3\; F$ and $\hat{R}_L = 10\;\Omega$ and the unknown deviations belonging to some known ranges, \textit{i.e.}, $\mu_C \in [-0.2, +0.2]$ and $\mu_R \in [-5, +5]$. With this at hand, the next step is to characterize how these unknown deviations affect the solution of the regulator equations $\Delta$. 
It is possible to show that the trajectories of $\Delta^{*}$ subject to any unknown derivations $\mu_C$ and $\mu_L$ within their range can be written as
\begin{equation}\label{equ:exampleDeltaMu}
\begin{aligned}
    \Delta^{*}(t, \mu) &= \left[\begin{array}{lll}
    \delta_1 & \mu_1\delta_2(t)+\mu_3 & \mu_2\delta_3(t)+\mu_4
    \end{array}\right], \\
    &= \Delta(t)\left[\begin{array}{ccc}
    1 & \mu_3/ \delta_1 & \mu_4/ \delta_1 \\
    0 & \mu_1 & 0 \\
    0 & 0 & \mu_2
    \end{array}\right] = \Delta(t)U^{\prime}(\mu).
\end{aligned}
\end{equation}
with $\Delta(t) = [\delta_1, \delta_2(t), \delta_3(t)]$ being the solution to the regulator equations with nominal values $C_L = 0.3\;F$ and $R_L = 10 \;\Omega$ (shown in Fig.~\ref{fig:Regulator_Equations}) and $\mu_1$, $\mu_2$, $\mu_3$, and $\mu_4$ unknown coefficients. Note that $\delta_1$ is a constant. By~(\ref{equ:exampleDeltaMu}), the problem cannot be solved by the normal error-feedback controller as $U^{\prime}(\mu)$ and $\Lambda$ are not commutative. However, this can be solved by either an augmentation-based robust regulator or an immersion-based robust regulator. Here, we implement both regulators for comparison. The real values of the load components are randomly selected as $C_L = 0.156\;F$ and $R_L = 7.891\;\Omega$. The initial values of $x$ and $\omega$ are randomly selected as $x_0 = [1.0562, \; -0.2586]^{\top}$ and $\omega_0 = [\hat{d}_0,\; \hat{r}_0,\; 0]^{\top}$, with $\hat{d}_0 = 0.8481$ and $\hat{r}_0 = -0.5202$.

\begin{figure}[t]
	\begin{centering}
		\includegraphics[width=\linewidth]{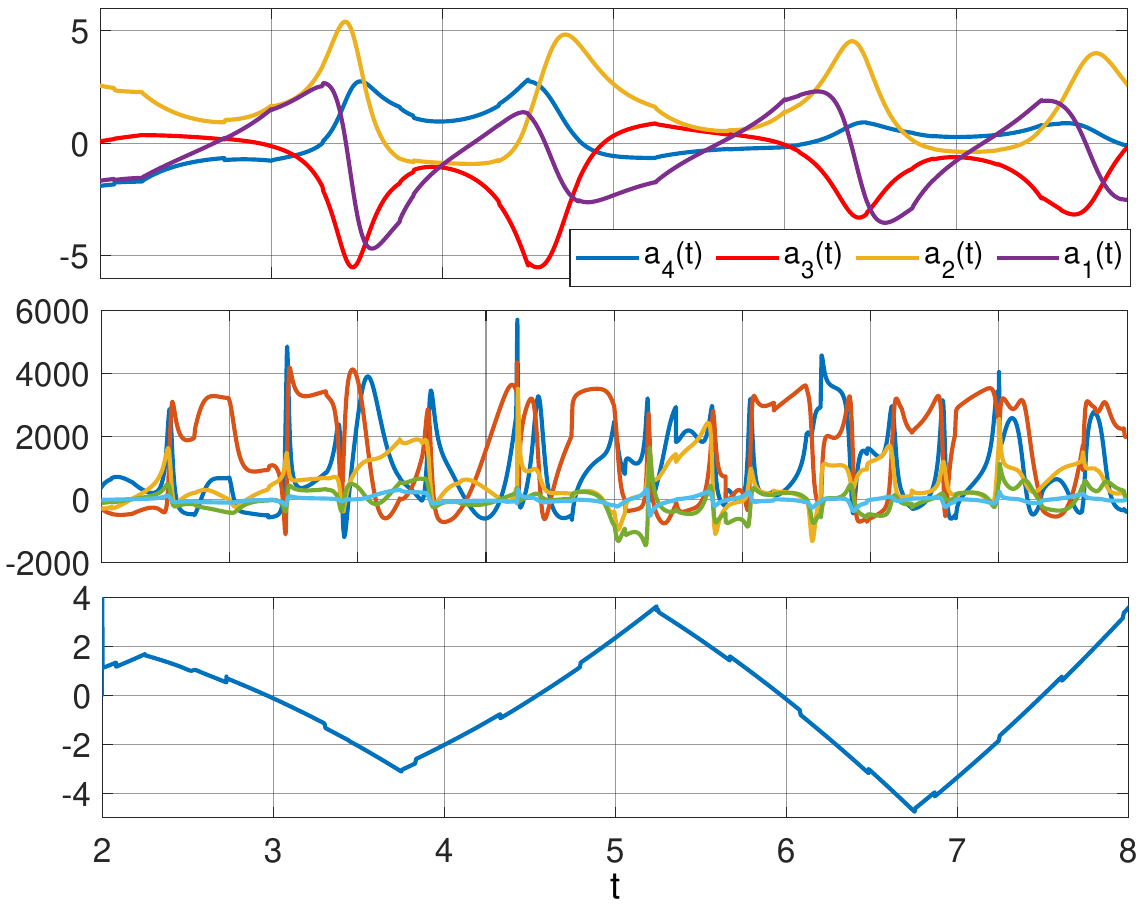}
		\caption{\textbf{Top:} Time histories of functions $a_1(t)$, $a_2(t)$, $a_3(t)$, $a_4(t)$. \textbf{Middle:} Time history of $H_{\mathrm{im}}(t)$ for the canonical realization. \textbf{Bottom:} Time history of the input $u(t)$ generated by the robust immersion-based regulator.}
	\label{fig:Immersion_Regulator}
	\end{centering}
\end{figure}
\begin{figure}[t]
	\begin{centering}
		\includegraphics[width=\linewidth]{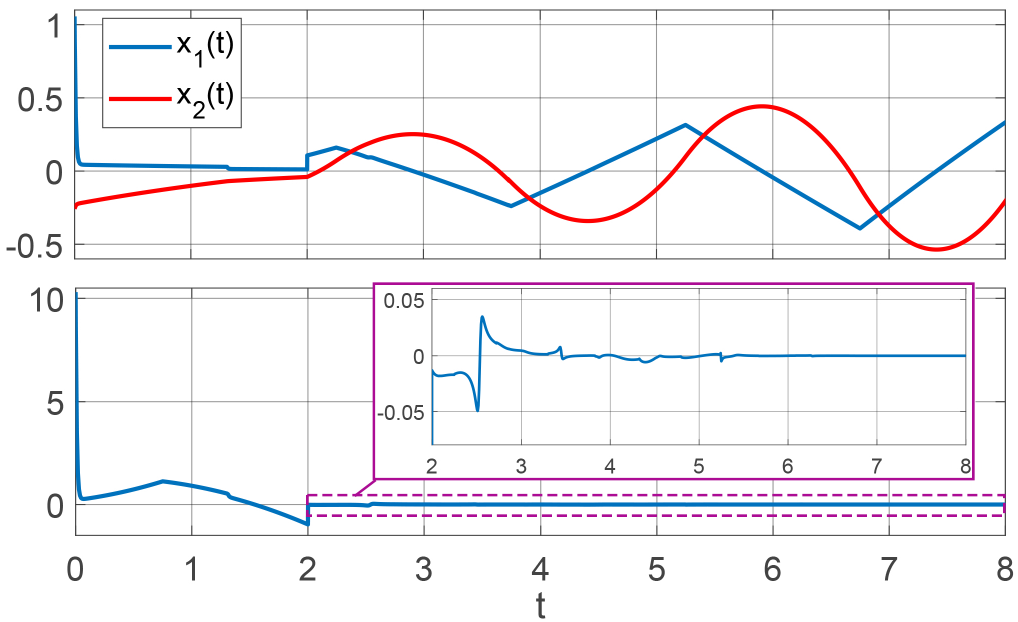}
		\caption{\textbf{Top:} Time history of system state $x(t)$. \textbf{Bottom:} Time history of the regulation error $e(t)$.}
	\label{fig:Immersion_Result}
	\end{centering}
\end{figure}

The design of an augmentation-based regulator follows the method of reducing the regulator dimension by finding the matrix $L$ satisfying~(\ref{equ:AugMethodRedOrder}). In this example, $l^{\prime} = 5$ so the dimension of the internal model is $\nu l^{\prime} = 15$. Here we set the dimension of the regulator (\ref{equ:errorFdReguCanonical}) as $q = 16$ (with $m=q$) as the first value above $\nu l^{\prime}$ for which the solution $\Pi_M$ in~(\ref{equ:IMCanoniParamMoment}) satisfies $\operatorname{rank}(\Pi_M) = s_2 = \nu l^{\prime}$. The pair $(F_{\mathrm{im}}, G_{\mathrm{im}})$ is constructed in the controllable canonical form with the eigenvalues of $F_{\mathrm{im}}$ randomly selected as $-0.5 \pm 0.5i$, $-1.2 \pm i$, $-1.5 \pm 1.2i$, $-2.0 \pm 1.8i$, $-2.5 \pm 2.0i$, $-0.3 \pm 1.5i$, $-0.8 \pm 0.4i$, and $-1.0 \pm 1.2i$. The high-gain value for stabilisation is selected as $k = 100$. By implementing the augmentation-based robust regulator, Fig.~\ref{fig:Augmentation_Result} (a) displays the time history of system state $x(t)$, while Fig.~\ref{fig:Augmentation_Result} (b) shows the time history of $H_{\mathrm{im}}(t)$ of the canonical realization given in Proposition~\ref{prop:pseudoInverse}.
Fig.~\ref{fig:Augmentation_Result} (c) and (d) display the time histories of the input signal $u(t)$ generated by the augmentation-based regulator and of the resulting regulation error $e(t)$, respectively. The figure implies that the augmentation-based regulator achieves the asymptotic convergence of the regulation error to zero despite the state trajectory diverging with time due to the divergence of the exogenous signal $r_f(t)$. The mapping $H_{\mathrm{im}}$ is shown to be bounded, while the input signal $u(t)$ is diverging and piecewise continuous.

We finally design a robust immersion-based regulator for solving the regulation problem with the same $C_L = 0.156\;F$ and $R_L = 7.891\;\Omega$. It is possible to show
that $R^{*}(t, \mu)$ satisfies $R^{*}(t, \mu) = \Delta(t)\Lambda(t, 0)U_R(\mu)$ with $U_R(\mu) = \operatorname{diag}(1, \mu_5, \mu_6)$ for some unknown constants $\mu_5$, $\mu_6$. This implies that an explicit immersion can be constructed if we can find bounded piecewise continuous functions $a_1(t), a_2(t), \ldots, a_{d}(t)$ solving~(\ref{equ:immersion}) with $R^{*}(t, \mu)$ replaced by $\Delta(t)\Lambda(t, 0)$. These functions are determined by following the method discussed in Example~\ref{eg:Immersion} with $d = 4$ and $\hat{t} = 2$. Then we set the regulator dimension $q = 5$. The eigenvalues of $F_{\mathrm{im}}$ are arbitrarily selected as $-5 \pm 4i$, $-3 \pm 2i$, and $-6$ with the pair $(F_{\mathrm{im}}, G_{\mathrm{im}})$ in controllable canonical form. The high-gain value for stabilisation is still $k = 100$. Since $\hat{t} = 2$, we set the input $u(t) = 0$ when $t < 2$. Fig.~\ref{fig:Immersion_Regulator} displays the time histories of functions $a_1(t)$, $a_2(t)$, $a_3(t)$, $a_4(t)$ that ensure the explicit immersion by~(\ref{equ:immersion}) (top graph), the time history of $H_{\mathrm{im}}(t)$ for the canonical realisation given in Corollary~\ref{corol:pseudoInverse} (middle graph), and the time history of the resulting input signal $u(t)$ (bottom graph). Fig.~\ref{fig:Immersion_Result} displays the time histories of the system state $x(t)$ (top graph) and the regulation error $e(t)$ (bottom graph). We can see from both figures that both the functions $a_i$, with $i = 1, 2, 3, 4$, and the mapping $H_{\mathrm{im}}$ are bounded. Meanwhile, an error-zeroing piecewise continuous input signal, similar to the input generated by the augmentation-based internal model, is successfully produced by this immersion-based internal model for $t > 2$. As a consequence, the regulation error asymptotically converges to zero when $t > 2$.

\section{Conclusion}\label{sec:concl}
In this article, we solved the output regulation problem for linear systems with non-smooth (possibly non-periodic) exogenous signals. In particular, we first solved the full-information problem by designing a state-feedback regulator based on new regulator equations. We then discussed the solvability of the regulator equations, showing that the solvability is related to the relative degree of the linear system, and proposing a new non-resonance condition. Based on the full-information solutions, we studied the error-feedback problem, where we designed an error-feedback regulator by finding the canonical realisation of the internal model, and we stabilised the closed-loop system under the minimum-phase assumption. After that, we focused on the internal model principle, and we proposed two methods of finding a robust internal model for the problem with model parameter uncertainties. We finally provided a circuit regulation example to illustrate the implementation of robust regulators when the problem involves parameter uncertainties. 

\section*{References}
\bibliographystyle{IEEEtran}
\bibliography{IEEEabrv,mybib}

\appendix

\subsection{Proof of Lemma~\ref{lem:moment}}\label{equ:lemMomentProof}

\begin{proof}
We first show that the family of matrix-valued functions $\Pi_{g}$ in~(\ref{equ:PIintegral}) parameterized in $\Pi_{g}(t_0) \in \mathbb{R}^{g \times \nu}$ is bounded and piecewise continuous. Under Assumption~\ref{asmp:exoProp}, Since $\Lambda$ is finite-time bounded and piecewise continuous while $B_g$ is bounded piecewise continuous, $\Pi_{g}(t)$ exists and is unique for any fixed $\Pi_{g}(t_0)$ and for any finite time~\cite[Theorem 3.1]{ref:khalil2002nonlinear}. Then to show the boundedness of $\Pi_{g}$, we need to prove that $\lim_{t \rightarrow+\infty} \Pi_{g}(t)$ is bounded. To this end, define $W_{\Lambda}(t, t_0) = \lim_{t_f \rightarrow+\infty} \Lambda(t, t_0)\Lambda(t_f, t_0)^{-1}$. Since $\Lambda$ is finite-time bounded and $\Lambda^{-1}$ is bounded, $W_{\Lambda}$ is finite-time bounded and $\lim_{t \rightarrow+\infty} W_{\Lambda}(t, t_0) = I_{\nu}$, implying that $W_{\Lambda}$ is bounded for all times. Meanwhile, as system~(\ref{equ:expStableSystem}) with $\omega \equiv 0$ is exponentially stable and $B_g$ is bounded for all times, there exist positive constants $\alpha$, $\beta$, $h_1$, $h_2$ such that $\|\Phi_{g}(t, \tau)\| \leq \alpha e^{-\beta(t - \tau)}$, $\| B_g(\tau) \| \leq h_1$, $\| W_{\Lambda}(\tau, t_0) \| \leq h_2$ for all $t \geq \tau \geq t_0$. Then with $\Lambda^{-1}$ bounded for all times, we have
\begin{equation}
\label{eq-limPI_g}
\begin{aligned}
    \lim_{t \rightarrow+\infty} \!\!\|\Pi_{g}(t)\| \!&= \!\!\!\lim_{t \rightarrow +\infty}\! \left\| \int_{t_0}^t \!\Phi_{g}(t, \tau) B_g(\tau) \Lambda(\tau, t_0) d \tau \Lambda(t, t_0)^{-1} \!\right\| \\
   &\leq \lim_{t \rightarrow +\infty} \int_{t_0}^{t} \alpha e^{-\beta(t-\tau)} h_1 h_2 d \tau \\
    &= \lim_{t \rightarrow +\infty} \left.\frac{\alpha h_1 h_2}{\beta}e^{\beta(\tau - t)} \right|_{\tau = t_0}^{t} = \frac{\alpha h_1 h_2}{\beta},
\end{aligned}
\end{equation}
yielding that $\Pi_{g}$ is bounded for all times. Meanwhile, note that because of the piecewise continuity of $\Lambda^{-1}$, $\Pi_{g}$ is also piecewise continuous. 

Consider now the unique solution of the closed-loop system~(\ref{equ:PsiSystem}) as
\begin{equation*}
x_{g}(t)=\Phi_{g}(t, t_0) x_{g}(t_0)+\!\!\int_{t_0}^t \Phi_{g}(t, \tau) B_g(\tau) \Lambda(\tau, t_0) \omega(t_0) d \tau.
\end{equation*}
Then for any pair of $\left((x_{g}({t_0}), \Pi_{g}(t_0)\omega(t_0)\right)$, tedious but straightforward calculations shows that $x_{g}(t) - \Pi_{g}(t) \omega(t) = \Phi_{g}(t, t_0)\left(x_{g}(t_0) - \Pi_{g}(t_0)\omega(t_0)\right)$.
As $\|\Phi_{g}(t, t_0)\| \leq \alpha e^{-\beta(t - t_0)}$ for all $t \geq t_0$ with some positive constants $\alpha$ and $\beta$, $\lim_{t \rightarrow+\infty}x_{g}(t)-\Pi_{g}(t) \omega(t)=\bm{0}_{g\times 1}$. Moreover, let $\Psi_{g}(t) = \Pi_{g}(t)\Lambda(t, t_0)$. By~(\ref{equ:PIintegral}), we have 
\begin{equation}\label{equ:Psi_g}
\Psi_{g}(t) \!=\! \Phi_{g}(t, t_0) \Psi_{g}(t_0) +\!\int_{t_0}^t \Phi_{g}(t, \tau) B_g(\tau) \Lambda(\tau, t_0) d \tau,
\end{equation}
which is absolutely continuous. Then since $A_g$ is bounded piecewise continuous and therefore Lebesgue integrable, Carathéodory's theorem implies that $\Psi_{g}$ in~(\ref{equ:Psi_g}) is the unique solution to the differential equation~(\ref{equ:PsiSystem}) with initial condition $\Psi_{g}(t_0) = \Pi_{g}(t_0)\Lambda(t_0, t_0)$ and satisfies $\lim_{t \rightarrow+\infty}x_{g}(t)-\Psi_{g}(t) \omega_0 = \lim_{t \rightarrow+\infty}x_{g}(t)-\Pi_{g}(t) \omega(t) = \bm{0}_{g\times 1}$.
\end{proof}

\subsection{Proof of Theorem~\ref{thm:zeroError}}
\label{AppB}

\begin{proof}
Under Assumptions~\ref{asmp:exoProp} and~\ref{asmp:systemProp}, if the input $u$ that ensures the satisfaction of \textbf{($\mathbf{S_F}$)} and \textbf{($\mathbf{R_F}$)} is finite-time bounded and piecewise continuous, then Lemma~\ref{lem:moment} and Theorem~\ref{thm:solutionWithD} implies that there exist a positive constant $\hat{t}$ and an initial condition $x(\hat{t})$ such that the finite-time bounded trajectory $x(t)$ results in $e(\hat{t}) = 0$ for all $t \geq \hat{t}$. Meanwhile, since $x$ and $u$ are finite-time bounded, $\dot{x}$ is finite-time bounded and $x$ is locally Lipschitz continuous~\cite[Chapter 3]{ref:khalil2002nonlinear}. This means that for any time constants $t_1$ and $t_2$ greater than $\hat{t}$, there exists a Lipschitz constant $L_{x}$ such that $\|x(t_1)-x(t_2)\| \leq L_{x}\|t_1-t_2\|$. Then the satisfaction of \textbf{($\mathbf{R_F}$)} yields that $e(t_1) = Cx(t_1) + Q\omega(t_1) = 0$ and $e(t_2) = Cx(t_2) + Q\omega(t_2) = 0$ for all $t \geq \hat{t}$. Subtracting these two equations, we get $C(x(t_1) - x(t_2)) + Q(\omega(t_1) - \omega(t_2)) = 0$, which yields
$\|Q(\omega(t_1)-\omega(t_2)\| = \|C(x(t_1) - x(t_2))\| \leq \|C\|\|x(t_1) - x(t_2)\| \leq \|C\|L_{x}\|t_1-t_2\|$. Therefore, $Q\omega(\cdot)$ is locally Lipschitz continuous for all $t \geq \hat{t}$.
\end{proof}

\ifthenelse{\boolean{TAC}}{}{

\subsection{Proof of Theorem~\ref{thm:solCondBoth}}\label{sec:Appsolvability}

In this subsection, for brevity, we set $t_0 = 0$ without loss of generality and let $\Lambda(t) = \Lambda(t, 0)$ with a slight abuse of notation. For the time being, suppose Assumption~\ref{asmp:pwDiff} holds. With a given initial time $\hat{t}$, define a group of matrix-valued functions $V_{j}(t) \in \mathbb{R}^{1 \times \nu}$ as
\begin{equation}\label{equ:functionsVj}
    V_{j}(t) := \sum_{i=1}^{j} \mathcal{I}^{[i]}_{\hat{t}} \left[CA^{i-1}P\Lambda(t)\right] + Q\Lambda(t)
\end{equation}
for $j = 1, 2, \cdots, n$, with $V_0(t) = Q\Lambda(t)$. Then define a sequence of positive integers $S_c = \{c_0, c_1, \cdots, c_n\}$ with each $c_{j}$ denoting the degree of smoothness of the corresponding matrix-valued function $V_{j}$, \textit{i.e.}, $V_{j} \in \mathcal{C}^{c_j}$. 
The idea behind the definition of the functions $V_{j}$ can be explained as follows: we cannot use differentiation as a tool to prove our claims because of the non-smoothness of $\Lambda$. Thus, we define auxiliary functions $V_{j}$ as ``integrals" that have sufficient smoothness to be differentiated as needed. As it will be clear later, the derivatives of the functions $V_{j}$ will play a role in the proof of the solvability condition.

\begin{lem}\label{lem:VjProperty}
    Consider the group of functions~(\ref{equ:functionsVj}). Let $j^{*} = c_n$. Suppose $j^{*} < n$. The following properties hold:
    \begin{itemize}
    \item[(i)] For all $j \in \{0, 1, \cdots, j^{*}\}$, $c_j \geq j$.
    \item[(ii)] For all $j \in \{j^{*}+1 \cdots, n\}$, $c_{j} = j^{*}$. 
\end{itemize}
\end{lem}
\begin{proof}
    We first prove (ii). Considering $j^{*} = c_n < n$, the definition in~(\ref{equ:functionsVj}) implies that $V_{n}(t) := \sum_{i=j^{*}+2}^{n} \mathcal{I}^{[i]}_{\hat{t}} \left[CA^{i-1}P\Lambda(t)\right] + V_{j^{*}+1}(t).$ Since the term $\sum_{i=j^{*}+2}^{n} \mathcal{I}^{[i]}_{\hat{t}}\left[CA^{i-1}P\Lambda(t)\right]$ is at least ($j^{*}+1$)-times differentiable, the degree of smoothness $c_n = j^{*}$ of function $V_n$ can only be induced by the term $V_{j^{*}+1}$. In other words, $V_n \in  \mathcal{C}^{c_n}$ with $c_n = j^{*} < n$ only if $c_{j^{*}+1} = j^{*} = c_n$, where $c_{j^{*}+1}$ is the degree of smoothness of $V_{j^{*}+1}$. Moreover, as $V_{j}(t):= \sum_{i=j^{*}+2}^{j} \mathcal{I}^{[i]}_{\hat{t}} \left[CA^{i-1}P\Lambda(t)\right] + V_{j^{*}+1}(t)$ for all $j^{*}+1 < j \leq n$ with the term $\sum_{i=j^{*}+2}^{j} \mathcal{I}^{[i]}_{\hat{t}}\left[CA^{i-1}P\Lambda(t)\right]$ at least ($j^{*}+1$)-times differentiable, the fact $c_{j^{*}+1} = j^{*}$ also yields that $c_{j} = c_{j^{*}+1} = j^{*}$ for all $j^{*}+1 < j \leq n$, \textit{i.e.}, the argument (ii) holds. The argument (i) can be proved in a similar manner. For any $j \in \{0, 1, \cdots, j^{*}\}$, we have $V_{j^{*}+1}(t) := \sum_{i=j+1}^{j^{*}+1} \mathcal{I}^{[i]}_{\hat{t}} \left[CA^{i-1}P\Lambda(t)\right] + V_{j}(t).$ Assume, by contradiction, that $V_{j} \in C^{c_j}$ with $c_j < j \leq j^{*}$. Since the term $\sum_{i=j+1}^{j^{*}+1} \mathcal{I}^{[i]}_{\hat{t}} \left[CA^{i-1}P\Lambda(t)\right]$ is at least $j$-times differentiable, $V_{j^{*}+1} \in C^{c_{j^{*}+1}}$ with $c_{j^{*}+1} = c_j < j^{*}$, which contradicts the equality $c_{j^{*}+1} = j^{*}$ proved before. Therefore, $c_j \geq j$ for any $j \in \{0, 1, \cdots, j^{*}\}$.
\end{proof}

Lemma~\ref{lem:VjProperty}, in short, reflects two facts: i) if the degree of smoothness, $c_n$, of the term $V_n$ satisfies $c_n < n$, this degree of smoothness coincides with the degree of smoothness of the ($c_n$+1)-th term in the sequence of functions $V_j$'s; ii) all functions $V_j$ with $j \leq c_n$ are at least $j$-times differentiable. These properties are instrumental in investigating the role of the relative degree of system~(\ref{equ:system}) in the solvability of the regulator equations~(\ref{equ:reguEquation}). With this consideration in mind, the following proposition establishes a necessary condition for the solvability of the regulator equations~(\ref{equ:reguEquation}) when $D = 0$.

\begin{prop}\label{prop:RDCondition}
Consider Problem~\ref{prob:ORFullInfo} driven by the control law~(\ref{equ:controlLaw}). Suppose $D = 0$ and Assumptions~\ref{asmp:exoProp},~\ref{asmp:systemProp}, and~\ref{asmp:pwDiff} hold. There exist bounded piecewise continuous $\Pi_x$ and $\Delta$ solving~(\ref{equ:reguEquation}) only if system (\ref{equ:system}) has a relative degree $r$ satisfying $r \leq j^{*} + 1$, where $j^{*}=c_n$.
\end{prop}

\begin{proof}
When $D = 0$, by Corollary~\ref{corol:equivDAE}, the existence of bounded piecewise continuous $\Pi_x$ and $\Delta$ solving~(\ref{equ:reguEquation}) is equivalent to the existence of a positive constant $\hat{t}$ and matrix-valued functions $\Psi_x(\cdot)$ and $\hat{\Delta}(\cdot)$ solving
\begin{equation}\label{equ:solutionDAEwithoutD}
    \begin{aligned}
    \dot{\Psi}_x(t) &= A\Psi_x^{*}(t) + (B\hat{\Delta}(t) + P)\Lambda(t), \\
    \bm{0}_{1 \times \nu} &= C \Psi_x^{*}(t) + Q\Lambda(t), \\
    \end{aligned}
\end{equation}
for all $t \geq \hat{t}$, with $\hat{\Delta}(\cdot)$ being bounded piecewise continuous. If $j^* \geq n$, then $r \leq n$ is surely smaller than $j^*+1$. Thus, consider the case $j^{*} < n$. Assume by contradiction that $r > j^{*} + 1$. By Lemma~\ref{lem:VjProperty}(ii), $c_{j^*+1}= j^*$, which means by definition that $V_{j^{*}+1} \in \mathcal{C}^{j^{*}}$ and, by (i), $V_{j}$ is at least $j$-times differentiable for all $j \leq j^{*}$. Then, note that the second equation in~(\ref{equ:solutionDAEwithoutD}) can be written as
$\bm{0}_{1 \times \nu} = C\Psi(t) + V_0^{(0)}(t)$,
and that the equation
$\bm{0}_{1 \times \nu} = CA\Psi(t) + CB\Delta(t) + V_1^{(1)}(t)$
is its first derivative. By iterating this for $j^{*}$ times we obtain
\begin{equation}\label{equ:solutionDAEwithoutDDiff}
    \bm{0}_{1 \times \nu} = CA^{j^{*}}\Psi_x(t) + CA^{j^{*}-1}B\Delta(t) + V_{j^{*}}^{(j^{*})}(t).
\end{equation}
Note that $V_{j^*}$ is $j^*$-times differentiable, and thus~(\ref{equ:solutionDAEwithoutDDiff}) is well defined. Then, we substitute the integral solution of the first equation in~(\ref{equ:solutionDAEwithoutD}) into~(\ref{equ:solutionDAEwithoutDDiff}), and we use the fact that $CA^i B=0$ for all $i\le j^*$ because $r>j^*+1$, obtaining
\begin{equation}\label{equ:solutionDAEwithoutDDiffCont}
    \bm{0}_{1 \times \nu}\! = \!\!\int_{\hat{t}}^{t} \!\!  CA^{j^{*}+1}\Psi_x(t) + CA^{j^{*}}\!\!P\Lambda(t) d\tau + \hat{G} + V_{j^{*}}^{(j^{*})}(t),
\end{equation}
for all $t \geq \hat{t}$, where $\hat{G} = CA^{j^{*}}\Psi_x(\hat{t})$ is a constant matrix. Now considering that $V_{j^{*}+1}(t) = V_{j^{*}}(t) + \mathcal{I}^{[j^{*}+1]}_{\hat{t}} \left[CA^{j^{*}}P\Lambda(t)\right]$, we have $V_{j^{*}+1}^{(j^{*})}(t) = V_{j^{*}}^{(j^{*})}(t) + \int_{\hat{t}}^{t} CA^{j^{*}}P\Lambda(t) d\tau$. Then~(\ref{equ:solutionDAEwithoutDDiffCont}) can be rewritten as
\begin{equation*}
    -\hat{G} =\!\! \int_{\hat{t}}^{t}\! CA^{j^{*}+1}\Psi_x(\tau) d\tau + V_{j^{*}+1}^{(j^{*})}(t),
\end{equation*}
Since $V_{j^{*}+1} \in \mathcal{C}^{j^{*}}$, $V_{j^{*}+1}^{(j^{*})} \in \mathcal{C}^{0}$. By the fact that $\int_{\hat{t}}^{t} CA^{j^{*}}\Psi_x(t) d\tau$ is continuously differentiable and cannot generate functions of class $\mathcal{C}^{0}$, the last equation is not solvable. Therefore, the regulator equations~(\ref{equ:reguEquation}) are solvable only if $r \leq j^* + 1$.
\end{proof}

\begin{remark}
Proposition~\ref{prop:RDCondition} discusses the relation between the non-smooth output dynamics of system~(\ref{equ:system}) induced by $\omega$ and the relative degree that is necessary for cancelling the resulting non-smoothness via an external finite-time bounded piecewise continuous input $u$. However, the requirement $r \leq j^{*} + 1$ may not sufficiently guarantee this cancellation. For example, if $V_{j^{*}+1}^{(j^{*})}$ is continuous with a derivative that at some time instants is not finite, then (\ref{equ:solutionDAEwithoutDDiffCont}) cannot be satisfied when $r = j^{*} + 1$ because $CA^{j^{*}-1}B = 0$ and $CA^{j^{*}}\Psi_x$  has always a finite-time bounded derivative under the control of a finite-time bounded input $u$.
\end{remark}

In fact, the case $j^{*} = c_n \geq 1$ implies that the regulation error $e$ of the linear system~(\ref{equ:system}) subject to the explicit generator~(\ref{equ:explicitGen}) with $u \equiv 0$ is continuously differentiable. This can be seen by the fact that $\dot{e}(t) = CAx(t) + \dot{V}_1(t)\omega_0$ where, by Lemma~\ref{lem:VjProperty}, $V_1$ is continuously differentiable. In the non-smooth case that we target (possibly time-varying triangular waves and square waves), the case $c_n \geq 1$ is less likely to happen. Therefore, in this section, we mainly focus on the case $r = 1$, which, by Proposition~\ref{prop:RDCondition}, is necessary for guaranteeing the solvability of the regulator equations for any $\Lambda$ under Assumptions~\ref{asmp:exoProp} and~\ref{asmp:pwDiff}.

Then, we study the solvability condition for the regulation problem under extra Assumptions~\ref{asmp:pwDiff} and~\ref{asmp:unitaryRD}. For the time being, suppose Assumption~\ref{asmp:pwDiff} holds. Then for any compact time interval $[t_a, t_b] \subset [t_0, +\infty)$, there exists a finite subdivision $t_a = \hat{t}_0 < \hat{t}_1 < \cdots <\hat{t}_{k-1} < \hat{t}_{k} = t_b$ of $[t_a, t_b]$  such that the function $Q\Lambda(t)$ is continuously differentiable in each subinterval $[\hat{t}_{j-1},\; \hat{t}_{j}]$ for any $j = 1, 2, \cdots, k$. Moreover, define a matrix-valued function $Q_{\Lambda}: \mathbb{R} \to \mathbb{R}^{1 \times \nu}$ such that $Q_{\Lambda}(t) := Q\dot{\Lambda}(t)\Lambda(t)^{-1}$ for all $\tilde{t} \in [\hat{t}_{j-1},\; \hat{t}_{j}) \subset [t_a, t_b]$ with $j = 1, 2, \cdots, n$ and for any subset $[t_a, t_b] \subset [t_0, +\infty)$. Then $Q_{\Lambda}\Lambda$ satisfies
\begin{equation*}
    Q\Lambda(t) = Q\Lambda(0) + \int_{0}^{t} Q_{\Lambda}(\tau)\Lambda(\tau) d\tau
\end{equation*}
and is bounded and piecewise continuous for all times under Assumption~\ref{asmp:pwDiff}. Now we are ready to check the solvability. For simplicity, let $b = CB$.


\begin{thm}\label{thm:SolCondWithoutD}
Consider Problem~\ref{prob:ORFullInfo} driven by control law~(\ref{equ:controlLaw}). Suppose Assumptions~\ref{asmp:exoProp},~\ref{asmp:systemProp},~\ref{asmp:pwDiff}, and~\ref{asmp:unitaryRD} hold. There exist bounded piecewise continuous matrices $\Pi_x$ and $\Delta$ solving the regulator equations (\ref{equ:reguEquation}) for any $P$ and $Q$ if and only if systems~(\ref{equ:system}) and~(\ref{equ:explicitGen}) are non-resonant.
\end{thm}
\begin{proof} 
By Corollary~\ref{corol:equivDAE}, we need to prove that, when system~(\ref{equ:system}) has a unitary relative degree, there exist a positive constant $\hat{t}$ and matrix-valued functions $\Psi_x(\cdot)$ and $\hat{\Delta}(\cdot)$ solving~(\ref{equ:solutionDAE}) for all $t \geq \hat{t}$ if and only if systems~(\ref{equ:system}) and~(\ref{equ:explicitGen}) are non-resonant. To show this, define a linear transformation $T := [T_1^{\top},\;T_2^{\top}]^{\top}$ with $T_2 = C$ and $T_1$ being selected such that $T$ is non-singular. By the change of variable $\Psi_x \mapsto \Theta = T\Psi_x := \operatorname{col}\left(\Theta_z, \Theta_y\right)$ with $\Theta_z \in \mathbb{R}^{(n-1) \times \nu}$ and $\Theta_y \in \mathbb{R}^{1 \times \nu}$,~(\ref{equ:solutionDAEwithoutD}) becomes
\begin{equation*}
    \begin{aligned}
    \dot{\Theta}_z(t) \!&=\! A_{11}\Theta_z(t)\! +\! A_{12}\Theta_y(t)\! +\! P_{1}\Lambda(t), \\ 
    \dot{\Theta}_y(t) \!&=\! A_{21}\Theta_z(t)\! +\! A_{22}\Theta_y(t)\! +\! P_{2}\Lambda(t)\! +\! b\hat{\Delta}(t)\Lambda(t),  \\
    \bm{0}_{1 \times \nu} \!&= \Theta_y(t) + Q\Lambda(t), 
    \end{aligned}
\end{equation*}
with
\begin{equation*}
    T A T^{-1} =\left[\begin{array}{ll}
        A_{11} & A_{12} \\
        A_{21} & A_{22}
    \end{array}\right], \quad 
    T P =  \left[\begin{array}{l}
        P_{1} \\
        P_{2}
    \end{array}\right].
\end{equation*}
Since $\Theta_y(t) + Q\Lambda(t) = \bm{0}_{1 \times \nu}$ for all $t \geq \hat{t}$, the change of variable $\Theta \mapsto \bar{\Pi}_x := \Theta\Lambda^{-1} = \operatorname{col}\left(\bar{\Pi}_z, \bar{\Pi}_y\right)$ yields
\begin{subequations}\label{equ:solvabCheckNFPi}
    \begin{align}
    \!\!\!\!\bar{\Pi}_z(t) \!&=\! \left( \!e^{A_{11} t} \bar{\Pi}_z(\hat{t}) \!+ \!\!\int_{\hat{t}}^t e^{A_{11}(t-\tau)} G_1 \Lambda(\tau) d \tau \!\! \right) \! \Lambda(t)^{-1}\!, \label{equ:solvabCheckNFIntDymPi}\\ 
    \!\!\!\!Q\Lambda(t) \!&=\! Q\Lambda(\hat{t}) \!+ \!\!\int_{\hat{t}}^t \!(A_{21}\bar{\Pi}_z(\tau) + G_2 + b\hat{\Delta}(\tau))\Lambda(\tau) d\tau, \label{equ:solvabCheckOutPi}
    \end{align}
\end{subequations}
for all $t \geq \hat{t}$ with $G_1 = P_1 - A_{12}Q$, $G_2 = P_{2} - A_{22}Q$, and $\bar{\Pi}_y(t) = Q\Lambda(t)$. By Lemma~\ref{lem:moment},~Theorem~\ref{thm:solutionWithD}, and Corollary~\ref{corol:equivDAE}, to have bounded piecewise continuous matrix-valued functions $\Pi_x = \Psi_x\Lambda^{-1}$ and $\Delta = \hat{\Delta}$ solving the regulator equations~(\ref{equ:reguEquation}), $\bar{\Pi}_z$ and $\hat{\Delta}$, that solve~(\ref{equ:solvabCheckNFPi}) with $\Pi_z$ obtained after the transformation $\bar{\Pi}_x = T\Pi_x$, should be bounded piecewise continuous for all $t \geq \hat{t}$. Note that under Assumptions~\ref{asmp:exoProp} and~\ref{asmp:pwDiff}, there always exists a bounded piecewise continuous $\hat{\Delta}$ solving~(\ref{equ:solvabCheckOutPi}) if a bounded $\bar{\Pi}_z$ solving~(\ref{equ:solvabCheckNFIntDymPi}) exists for all $t \geq \hat{t}$. For example, with a bounded $\bar{\Pi}_z$, there exists a bounded piecewise continuous $\hat{\Delta}$ expressed by
\begin{equation*}
    \hat{\Delta}(t) = b^{-1}\left(Q_\Lambda(t) - G_{2} - A_{21}\bar{\Pi}_z(t)\right),
\end{equation*}
which solves~(\ref{equ:solvabCheckOutPi}). Then, the existence of bounded piecewise continuous $\Pi_x$ and $\Delta$ solving the regulator equations~(\ref{equ:reguEquation}) comes down to the existence of the bounded solution $\bar{\Pi}_z$ in~(\ref{equ:solvabCheckNFIntDymPi}) with an arbitrary initial condition $\bar{\Pi}_z(\hat{t}) \in \mathbb{R}^{(n-1) \times \nu}$ for all $t \geq \hat{t}$. When Assumptions~\ref{asmp:exoProp} and~\ref{asmp:pwDiff} hold, the vectorization of~(\ref{equ:solvabCheckNFIntDymPi}) yields that a bounded solution $\bar{\Pi}_z$ exists for any $\bar{\Pi}(\hat{t})$ and any $P$ and $Q$ if and only if the non-smooth non-resonance condition is satisfied.
\end{proof}

\begin{remark}
When $\Lambda = e^{St}$, 
the expression~(\ref{equ:solvabCheckNFIntDymPi}) becomes
\begin{equation}\label{equ:solvabCheckNFIntDymPiSylvester}
    \bm{0}_{(n-r)\times \nu} = A_{11}\bar{\Pi}_z - \bar{\Pi}_{z}S + G_1,
\end{equation}
whose solvability is guaranteed by the traditional non-resonance condition. In this case, the boundedness of $\Omega$ in~(\ref{equ:nonSmNRC}) is only a sufficient but not necessary condition for the existence of the matrix $\bar{\Pi}_{z}$ solving~(\ref{equ:solvabCheckNFIntDymPiSylvester}), see~\cite{ref:bhatia1997and} for more discussion of the integral-from solution of the Sylvester equations.
\end{remark}

Finally, we consider the regulation problem with $D \neq 0$, in which case the relative degree of system~(\ref{equ:system}) is zero and Assumptions~\ref{asmp:pwDiff} and~\ref{asmp:unitaryRD} can be omitted. Solvability of the regulator equations (\ref{equ:reguEquation}) can be studied in a similar way as in Theorem~\ref{thm:SolCondWithoutD}.

\begin{thm}\label{thm:solCondWithD}
Consider Problem~\ref{prob:ORFullInfo} driven by control law~(\ref{equ:controlLaw}). Suppose $D \neq 0$ and Assumptions~\ref{asmp:exoProp} and~\ref{asmp:systemProp} hold. There exist bounded piecewise continuous matrices $\Pi_x$ and $\Delta$ solving the regulator equations~(\ref{equ:reguEquation}) for any $P$ and $Q$ if and only if systems~(\ref{equ:system}) and~(\ref{equ:explicitGen}) are non-resonant.
\end{thm}
\begin{proof}
By Corollary~\ref{corol:equivDAE} with the substitutions $\hat{\Delta} = -D^{-1}(C\Psi_x\Lambda^{-1} + Q)$ and $\hat{\Pi}_z = \Psi_x\Lambda^{-1}$, the regulator equations~(\ref{equ:reguEquation}) is solvable for any $P$ and $Q$ if and only if there exist a positive constant $\hat{t}$ and a bounded piecewise continuous matrix-valued function $\hat{\Pi}_z(t) \in \mathbb{R}^{n \times \nu}$ of the form
\begin{equation*}
    \hat{\Pi}_z(t) \!=\! \left( \!e^{A_{\pi} t} \hat{\Pi}_z(\hat{t}) \!+ \!\!\int_{\hat{t}}^t e^{A_{\pi}(t-\tau)} P_{\pi} \Lambda(\tau) d \tau \!\! \right) \! \Lambda(t)^{-1}\!, 
\end{equation*}
for all $t \geq \hat{t}$ and any initial condition $\hat{\Pi}_z(\hat{t})$, where $A_{\pi} =  A - BCD^{-1}$ and $P_{\pi} = P - BQD^{-1}$. Under Assumption~\ref{asmp:exoProp}, this boundedness is satisfied for any $P$ and $Q$ if and only if systems~(\ref{equ:system}) and~(\ref{equ:explicitGen}) are non-resonant.
\end{proof}

Rewriting Theorems~\ref{thm:SolCondWithoutD} and \ref{thm:solCondWithD} together yields Theorem~\ref{thm:solCondBoth}, concluding the proof of that theorem.

}



\begin{IEEEbiography}[{\includegraphics[width=1in,height=1.25in,clip,keepaspectratio]{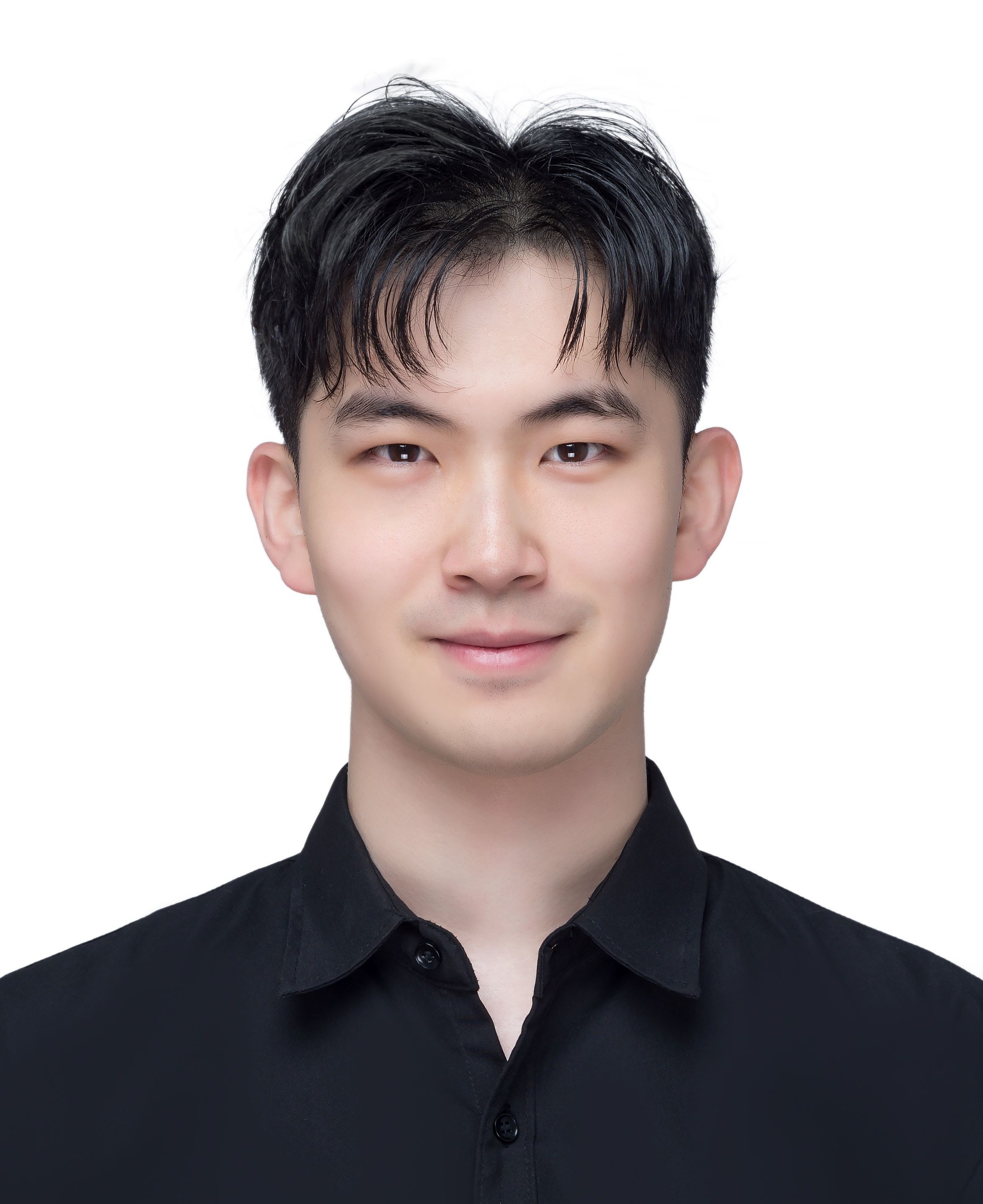}}]{Zirui Niu} (Graduate Student Member, IEEE)
was born in Shandong, China, in 1997. He received the B.Eng. (Hons) degree in electrical and electronic engineering from the University of Liverpool, UK, in 2020, and the M.Sc. degree in control systems from Imperial College London, UK, in 2021. Since 2022, he has been pursuing a Ph.D. in control theory at Imperial College London, UK, supported by the CSC-Imperial Scholarship. His research interests include output regulation, adaptive control, hybrid systems, model reduction, and robust control. He was the recipient of the MSc Control Systems Outstanding Achievement Prize (2021) and the Hertha Ayrton Centenary Prize (2021) for outstanding performance in MSc Control Systems and the outstanding master’s thesis in the Electrical and Electronic Engineering Department, respectively. 
\end{IEEEbiography}

\begin{IEEEbiography}
[{\includegraphics[width=1in,height=1.25in,clip,keepaspectratio]{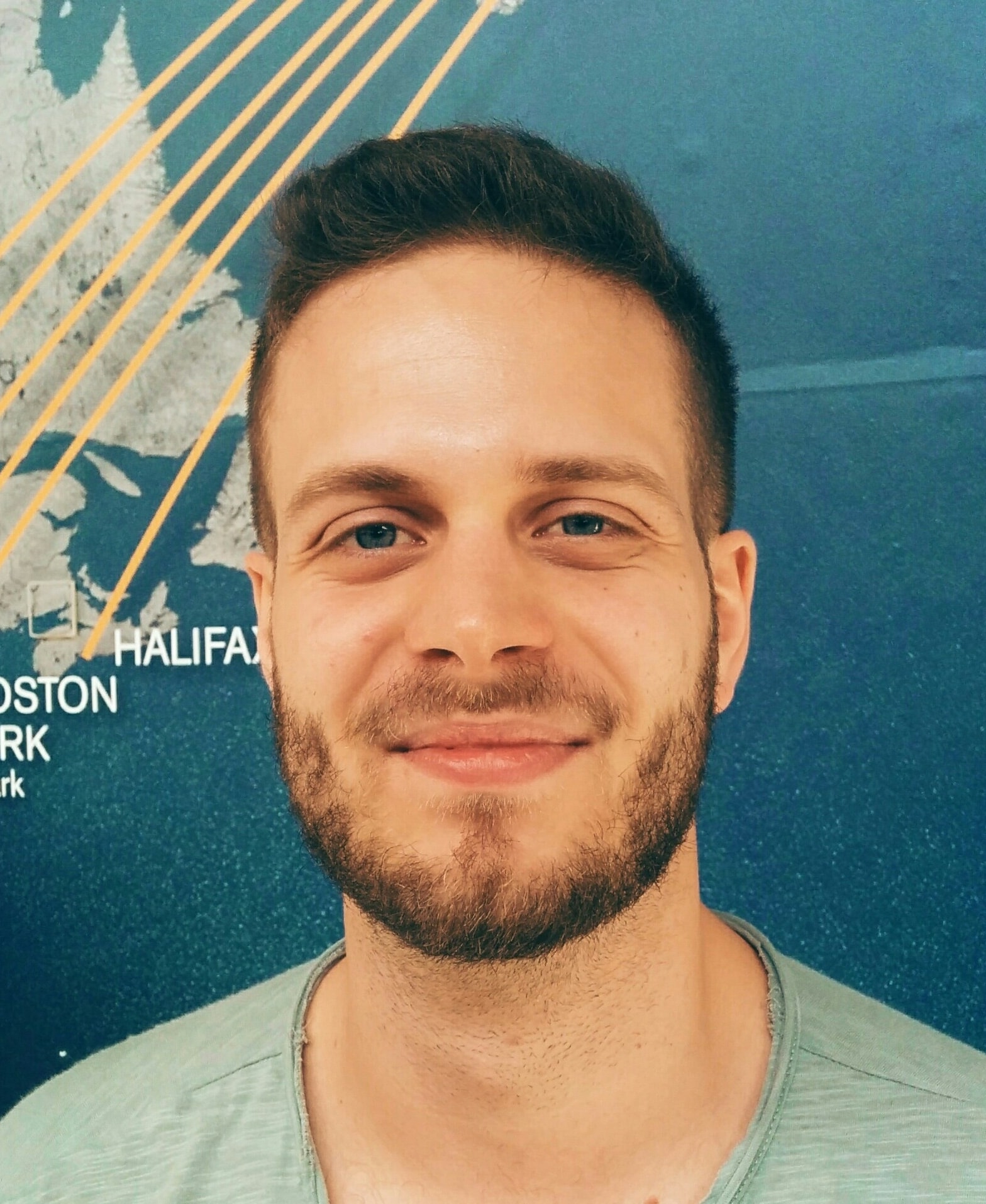}}]{Daniele Astolfi}
received the B.S. and M.S. degrees in
automation engineering from the University of Bologna,
Italy, in 2009 and 2012, respectively. He obtained a joint
Ph.D. degree in Control Theory from the University of
Bologna, Italy, and from Mines ParisTech, France, in 2016.
In 2016 and 2017, he has been a Research Assistant 
at the University of Lorraine (CRAN), Nancy, France.
Sinah ce 2018, he is a CNRS Researcher at 
LAGEPP, Lyon, France. 
His research interests include observer design, feedback
stabilization and output regulation  for nonlinear systems,
networked control systems, hybrid systems, and multi-agent systems.
 He serves as an associate 
editor of the IFAC journal Automatica since 2018 and European Journal of Control since 2023.
He was a recipient of the 2016 Best Italian Ph.D. Thesis Award in Automatica given by SIDRA.
\end{IEEEbiography}

\begin{IEEEbiography}
[{\includegraphics[width=1in,height=1.25in,clip,keepaspectratio]{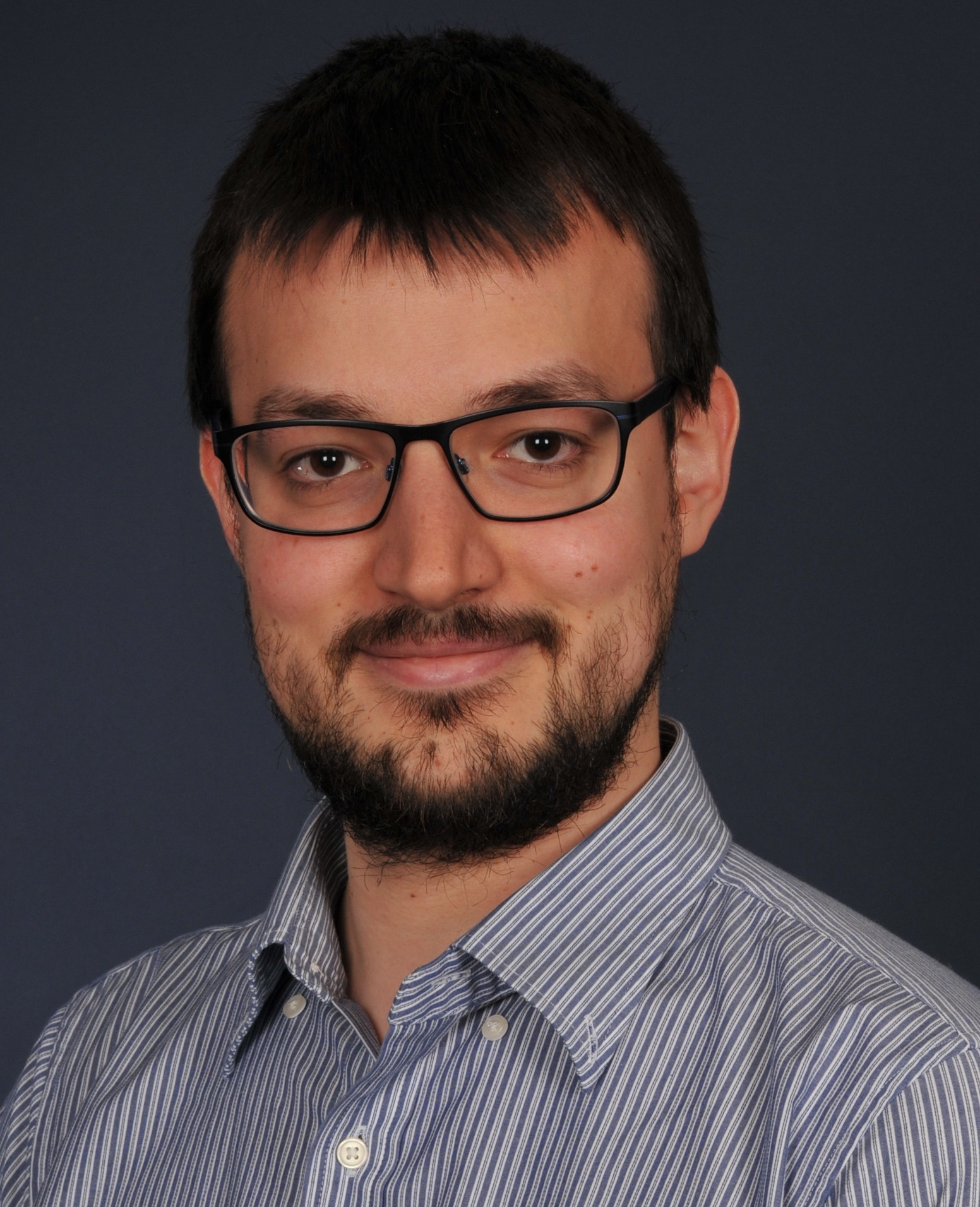}}]{Giordano Scarciotti}
Giordano Scarciotti (Senior Member, IEEE) was born in Frascati (Rome), Italy, in 1988. He received his B.Sc. and M.Sc. degrees in Automation Engineering from the University of Rome “Tor Vergata”, Italy, in 2010 and 2012, respectively. In 2012 he joined the Control and Power Group, Imperial College London, UK, where he obtained a Ph.D. degree in 2016. He also received an M.Sc. in Applied Mathematics from Imperial College in 2020. He is currently a Senior Lecturer in the Control and Power Group. His current research interests are focused on analysis and control of uncertain systems, model reduction and optimal control. He was a visiting scholar at New York University in 2015 and at University of California Santa Barbara in 2016 and a Visiting Fellow of Shanghai University in 2021-2022. He is the recipient of an Imperial College Junior Research Fellowship (2016), of the IET Control \& Automation PhD Award (2016), the Eryl Cadwaladr Davies Prize (2017) and an ItalyMadeMe award (2017). He received the IEEE Transactions on Control Systems Technology Outstanding Paper Award (2023). He is a member of the EUCA and IEEE CSS Conference Editorial Boards, and of the IFAC and IEEE CSS Technical Committees on Nonlinear Control Systems. He has served in the International Program Committees of multiple conferences, he is Associate Editor of Automatica and an Editor-at-Large of IEEE CDC 2024. He was the National Organizing Committee Chair for the EUCA European Control Conference 2022 and of the 7th IFAC Conference on Analysis and Control of Nonlinear Dynamics and Chaos 2024, and the Invited Session Chair for IFAC Symposium on Nonlinear Control Systems (NOLCOS) 2022. He is the Editor of the IFAC NOLCOS 2025. 
\end{IEEEbiography}

\end{document}